\documentclass{article}

\usepackage[margin = 1in]{geometry}
\usepackage{setspace}

\usepackage{waingarten}
\usepackage{thmtools}
\usepackage{thm-restate}
\usepackage{environ}

\providecommand{\email}[1]{\href{mailto:#1}{\nolinkurl{#1}}}

\usepackage{caption}
\usepackage[british]{babel}
\usepackage[applemac]{inputenc}
\usepackage{amsmath}
\usepackage{mathtools}
\usepackage{amssymb}
\usepackage{amsthm}
\usepackage{mathrsfs}
\usepackage{xcolor}
\usepackage{enumitem}
\usepackage{geometry}
\usepackage{hyperref}
\usepackage[nameinlink]{cleveref}
\usepackage{framed}
\usepackage{bm}
\usepackage{bbm}
\usepackage[square,sort,comma,numbers]{natbib}
\usepackage{caption}

\usepackage{xpatch}
\xpatchcmd{\proof}{\itshape}{\normalfont\proofnamefont}{}{}
\newcommand{\proofnamefont}{}
\renewcommand{\proofnamefont}{\bfseries}

\newcommand\rest[2]{{
	\left.\kern-\nulldelimiterspace 
  	#1 
  	\vphantom{|} 
  	\right|_{#2} 
}}

\newcommand{\comment}[1]{\begin{framed} #1 \end{framed}}

\newcommand{\floor}[1]{\lfloor #1 \rfloor}
\newcommand{\ceil}[1]{\lceil #1 \rceil}

\newcommand{\FindMon}{\texttt{Find-Monotone}}
\newcommand{\SampleSuffix}{\texttt{Sample-Suffix}}
\newcommand{\FindGoodSplit}{\texttt{Find-Good-Split}}
\newcommand{\FindWithinInterval}{\texttt{Find-Within-Interval}}

\title{Optimal Adaptive Detection of Monotone Patterns}
\author{Omri Ben-Eliezer\thanks{Tel-Aviv University, email: \email{omrib@mail.tau.ac.il}} 
\and 	Shoham Letzter\thanks{ETH Institute for Theoretical Studies, ETH Zurich, email: \email{shoham.letzter@eth-its.ethz.ch}}
\and 	Erik Waingarten\thanks{Columbia University, email: \email{eaw@cs.columbia.edu}}
}
\date{}

\begin{document}
\begin{titlepage}
\clearpage
\maketitle
\thispagestyle{empty}
\begin{abstract}
	\setlength{\parskip}{\medskipamount}
	\setlength{\parindent}{0pt}
We investigate adaptive sublinear algorithms for detecting monotone patterns in an array. Given fixed $2 \leq k \in \N $ and $\eps > 0$, consider the problem of finding a length-$k$ increasing subsequence in an array $f \colon [n] \to \mathbb{R}$, provided that $f$ is $\eps$-far from free of such subsequences. Recently, it was shown that the non-adaptive query complexity of the above task is $\Theta((\log n)^{\lfloor \log_2 k \rfloor})$. In this work, we break the non-adaptive lower bound, presenting an adaptive algorithm for this problem which makes $O(\log n)$ queries. This is optimal, matching the classical $\Omega(\log n)$ adaptive lower bound by Fischer [2004] for monotonicity testing (which corresponds to the case $k=2$), and implying in particular that the query complexity of testing whether the longest increasing subsequence (LIS) has constant length is $\Theta(\log n)$.
\end{abstract}
\end{titlepage}


%
%

\section{Introduction}


For an integer $k \in \N$ and a function (or sequence) $f \colon [n] \to \R$, a \emph{length-$k$ monotone subsequence of $f$} is a tuple of $k$ indices, $(i_1, \dots, i_k) \in [n]^k$, such that $i_1 < \dots < i_k$ and $f(i_1) < \dots < f(i_k)$. More generally, for a permutation $\pi \colon [k] \to [k]$, a \emph{$\pi$-pattern of $f$} is given by a tuple of $k$ indices $i_1 < \dots < i_k$ such that $f(i_{j_1}) < f(i_{j_2})$ whenever $j_1, j_2 \in [k]$ satisfy $\pi(j_1) < \pi(j_2)$. A sequence $f$ is $\pi$-free if there are no subsequences of $f$ with order pattern $\pi$. Pattern avoidance and detection in an array is a central problem in theoretical computer science and combinatorics, dating back to the work of Knuth \cite{K68} (from a computer science perspective), and Simion and Schmidt \cite{SS85} (from a combinatorics perspective); see also the survey \cite{V15b}. Studying the computational problem from a \emph{sublinear} algorithms perspective, Newman, Rabinovich, Rajendraprasad, and Sohler~\cite{NRRS17} initiated the study of property testing for forbidden order patterns in a sequence. For a fixed $k \in \N$ and a pattern $\pi$ of length $k$, we want to test whether a function $f\colon [n] \to \R$ is $\pi$-free or $\eps$-far from $\pi$-free.\footnote{A function $f$ is $\eps$-far from $\pi$-free if any $\pi$-free function $g$ differs on a $\eps n$ inputs.} They explicitly considered the monotone case as a particularly interesting instance;
monotone patterns are naturally connected to monotonicity testing and the longest increasing subsequence 
and can shine new light on these classic problems. Note that being free of length-$k$ monotone increasing subsequences is equivalent, as a simple special case of Dilworth's theorem \cite{Dilworth1950}, 
to being \emph{decomposable} into $k-1$ monotone non-increasing subsequences.
The algorithmic task, which is the subject of this paper, is the following:
\begin{quote}\itshape 
	For $2 \leq k \in \N$ and $\eps > 0$, design a randomized algorithm that, given query access to a function $f \colon [n] \to \R$, distinguishes with probability at least $9/10$ between the case that $f$ is free of length-$k$ monotone subsequences and the case that it is $\eps$-far from free of length-$k$ monotone subsequences.
\end{quote}
This paper gives an algorithm with optimal dependence in $n$ for solving the above problem. We state the main theorem next, and discuss connections to monotonicity testing and LIS shortly after.

\begin{theorem}\label{thm:intro-ub}
	Fix $k \in \N$. For any $\eps > 0$, there exists an algorithm that, given query access to a function $f \colon [n] \to \R$ which is $\eps$-far from $(12\dots k)$-free, outputs a length-$k$ monotone subsequence of $f$ with probability at least $9/10$, with query complexity and running time\footnote{Generally, along the context of the introduction, we allow the $\poly(1/\eps)$ term to depend on $k$; the precise bound we obtain here is $\left((k\log(1/\eps))^k (1/\eps)\right)^{O(k)} \cdot \log n$. See Lemma \ref{lem:alg-complexity} for more details.} of $\poly(1/\eps) \cdot \log n$.
\end{theorem}

The algorithm underlying Theorem~\ref{thm:intro-ub} is \emph{adaptive}\footnote{An algorithm is \emph{non-adaptive} if its queries do not depend on the answers to previous queries, or, equivalently, if all queries to the function can be made in parallel. Otherwise, if the queries of an algorithm may depend on the outputs of previous queries, then the algorithm is \emph{adaptive}.} and solves the testing problem with \emph{one-sided error},\footnote{An algorithm for testing property $\calP$ has \emph{one-sided error} if it has perfect completeness, i.e., it always outputs ``yes'' if $f \in \calP$; otherwise, the algorithm is said to have two-sided error.} since a length-$k$ monotone subsequence is evidence for not being $(12\dots k)$-free. The algorithm improves on a recent result of Ben-Eliezer, Canonne, Letzter and Waingarten \cite{BCLW19} who gave an algorithm for finding length-$k$ monotone patterns with query complexity $\poly(1/\eps) \cdot (\log n)^{\lfloor \log_2 k \rfloor}$, which in itself improved upon a $\poly(1/\eps) \cdot (\log n)^{O(k^2)}$ upper bound by Newman et al.~\cite{NRRS17}. The focus of \cite{BCLW19} was on \emph{non-adaptive} algorithms, and they gave a lower bound of $\Omega\!\left((\log n)^{\lfloor \log_2 k \rfloor}\right)$ queries for non-adaptive algorithms achieving one-sided error. Hence, Theorem~\ref{thm:intro-ub} implies a natural separation between the power of adaptive and non-adaptive algorithms for finding monotone subsequences.

Theorem~\ref{thm:intro-ub} is optimal, even among two-sided error algorithms. In the case $k=2$, corresponding to monotonicity testing, there is a $\Omega(\log n / \eps)$ lower bound\footnote{The precise lower bound is of the form $\Omega(\log (\eps n) / \eps)$, and is equivalent to the aforementioned one as long as, say, $\epsilon > n^{-0.99}$.} for both non-adaptive and adaptive algorithms \cite{EKKRV00, F04, CS14}, even with two-sided error. A simple reduction suggested in \cite{NRRS17} shows that the same lower bound (up to a multiplicative factor depending on $k$) holds for any fixed $k \geq 2$.
Thus, an appealing consequence of Theorem~\ref{thm:intro-ub} is that the natural generalization of monotonicity testing, which considers forbidden monotone patterns of fixed length longer than 2, does not affect the query complexity by more than a constant factor. Interestingly, Fischer \cite{F04} shows that for any adaptive algorithm for monotonicity testing on $f\colon [n] \to \R$ there is a non-adaptive algorithm at least as good in terms of query complexity (even if we only restrict ourselves to one-sided error algorithms). That is, adaptivity does not help at all for $k=2$. In contrast, the separation between our $O(\log n)$ adaptive upper bound and the $\Omega\!\left((\log n)^{\lfloor \log_2 k \rfloor}\right)$ non-adaptive lower bound of \cite{BCLW19} implies that this is no longer true for  $k \ge 4$.

As an immediate consequence, Theorem \ref{thm:intro-ub} gives an optimal testing algorithm for the longest increasing subsequence (LIS) problem in a certain regime.
The classical LIS problem asks one to determine, given a sequence $f \colon [n] \to \R$, what is the maximum $k$ for which $f$ contains a length-$k$ increasing subsequence. It is very closely related to other fundamental algorithmic problems in sequences, like the computation of edit distance, Ulam distance, or distance from monotonicity (for example, the latter equals $n$ minus the LIS length), and was thoroughly investigated from the perspective of sublinear-time algorithms \cite{PRR06, ACCL07, SS17, RSSS19} and streaming algorithms \cite{GJKK07,SW2007,GG10,SS13,EJ15,NSaks15}. In the property testing regime, the corresponding decision task is to distinguish between the case where $f$ has LIS length at most $k$ (where $k$ is given as part of the input) and the case that $f$ is $\eps$-far from having such a LIS length. Theorem \ref{thm:intro-ub} in combination with the aforementioned lower bounds (which readily carry on to this setting) yield a tight bound on the query complexity of testing whether the LIS length is a constant.
\begin{corollary}
	Fix $2 \leq k \in \N$ and $\eps > 0$. The query complexity of testing whether $f \colon [n] \to \R$ has LIS length at most $k$ is $\Theta(\log n)$.
\end{corollary}

\subsection{Related Work}
\label{subsec:related}
	Considering general permutations $\pi$ of length $k$ and \emph{exact} computation, Guillemot and Marx \cite{GM14} showed how to find a $\pi$-pattern in a sequence $f$ in time $2^{O(k^2 \log k)} n$, later improved by Fox \cite{Fox13} to $2^{O(k^2)} n$. In the regime $k = \Omega(\log n)$, an algorithm of Berendsohn, Kozma, and Marx \cite{BKM19} provides the state-of-the-art. 

	For \emph{approximate} computation of general patterns $\pi$, the works of \cite{NRRS17, BC18} investigate the query complexity of property testing for forbidden order patterns. When $\pi$ is of length $2$, the problem considered is equivalent to testing monotonicity, one of the most widely-studied problems in property testing, with works spanning the past two decades. Over the years, variants of monotonicity testing over various partially ordered sets have been considered, including 
	the line $[n]$ \cite{EKKRV00,F04,Bel18,PRV18,BE19}, the Boolean hypercube $\{0,1\}^d$~\cite{DGLRRS99,BBM12,BCGM12,CS13,CST14,CDST15,KMS15,BB15,CS16,CWX17, CS18}, and the hypergrid $[n]^d$~\cite{BRY14a,CS14, BCS18}. 
	We refer the reader to~\cite[Chapter~4]{G17} for more on monotonicity testing, and a general overview of the field of property testing (introduced in~\cite{RS96,GGR98}). 

	Understanding the power of adaptivity seems to be a notoriously difficult problem in property testing. In the context of testing for forbidden order patterns, non-adaptive algorithms are rather weak: the non-adaptive query complexity is $\Omega(n^{1/2})$ for all non-monotone order patterns \cite{NRRS17}, and as high as $n^{1 - 1/(k-\Theta(1))}$ for most order patterns of length $k$ \cite{BC18}. 
	Prior to our work (which shows separation between adaptive and non-adaptive algorithms for monotone patterns), the only case for which adaptive algorithms were known to outperform their non-adaptive counterparts have been for patterns of length $3$ in \cite{NRRS17}, and an intriguing conjecture from the same paper suggests that in fact, the query complexity for testing $\pi$-freeness is polylogarithmic in $n$ for \emph{any} fixed-length $\pi$ -- depicting an exponential separation from the non-adaptive case (for non-monotone patterns).




\subsection{Main Ideas and Techniques}
\label{subsec:techniques}

	We now describe the intuition behind the proof of Theorem~\ref{thm:intro-ub}. There are two main technical components: 1) a new structural result for functions $f\colon [n] \to \R$ with many length-$k$ monotone subsequences which strengthens a theorem of \cite{BCLW19}, and 2) new (adaptive) algorithmic components which lead to the $O(\log n)$-query algorithm. 
	We start by explaining the $(\log n)^{O(k^2)}$ upper bound of Newman et al.~\cite{NRRS17} and the structural decomposition of \cite{BCLW19}. 
	
	Fix $k \in \N$ and $\varepsilon > 0$, and suppose that $f \colon [n] \to \R$ is $\eps$-far from $(12\dots k)$-free, that is, $\eps$-far from free of length-$k$ increasing subsequences. Notice that $f$ must contain a collection $\calC$ of at least $\eps n / k$ pairwise-disjoint increasing subsequences of length $k$.\footnote{Otherwise, greedily eliminating these subsequences gives a $(12\dots k)$-free function differing in strictly less than $\eps n$ inputs.}
	For simplicity, consider $k=2$ first (which corresponds to the classical problem of monotonicity testing).
	For any $x < \ell < y \in [n]$, we say that $\ell$ \emph{cuts the pair $(x,y)$ with slack} if $x + (y-x)/3 \leq \ell \leq y - (y-x)/3$, or, informally, if $\ell$ lies roughly ``in the middle'' of $x$ and $y$. Additionally, the \emph{width} of the pair $(x,y)$ is $\floor{\log(y-x)}$. 
	Define the collection of copies from $\calC$ of width $w$ around $\ell$ by
	$$\calC_{\ell, w} = \{(x,y) \in \calC : \text{width}(x,y)=w,\  \text{$\ell$ cuts $(x,y)$ with slack}\}.$$
	Finally, the \emph{density} of copies from $\calC$ of width $w$ around $\ell$, and the total density of $\calC$ around $\ell$, are defined by
	$$\tau_\calC(\ell, w) = \frac{1}{2^w} \cdot |\calC_{\ell, w}| \hspace{0.2cm} ; \hspace{0.5cm} \tau_{\calC}(\ell) = \sum_{w=1}^{\log n} \tau_\calC(\ell, w).$$

	\paragraph{A polylogarithmic-query algorithm.}
		Fix a location $\ell \in [n]$ and a width $w \in [\log n]$, and consider drawing $\Theta(1 / \tau_{\calC}(\ell,w))$ indices from the interval $[\ell-2^w, \ell+2^w]$ uniformly at random, querying $f$ in all of these locations. Letting $m$ be the median of the set $\{f(x) : (x,y) \in \calC_{\ell, w}\}$, if we manage to query the ``1-entry'' $x$ of some $(x,y) \in \calC_{\ell, w}$ where $f(x) \leq m$, and the ``2-entry'' $y'$ of some $(x',y') \in \calC_{\ell, w}$ where $f(x') \geq m$, then $(x,y')$ would form a valid $(12)$-pattern, since $x < \ell < y'$ and $f(x) \leq m \leq f(x') \leq f(y')$. By definition, the number of entries $x$ as well as the number of entries $y'$ within $[\ell - 2^w, \ell+2^w]$ which may be sampled is at least $\Omega\left(\tau_{\calC}(\ell, w) \cdot 2^w\right)$. Therefore, with good probability, $\Theta(1 / \tau_{\calC}(\ell,w))$ uniform queries from the interval will hit at least one such $x$ and one $y'$, which would form the desired $(12)$-pattern.

		We claim that many values of $\ell \in [n]$ have some width $w \in [\log n]$ where the density $\tau_{\calC}(\ell, w)$ is large. First, a simple double counting argument shows $\mathbb{E}_{\ell \in [n]}[\tau_{\calC}(\ell)] = \Omega(\eps)$. On the other hand, $\tau_{\calC}(\ell, w) \leq O(1)$ for any width $w \in [\log n]$, and so $\tau_{\calC}(\ell) = O(\log n)$. Consequently, the probability that a random $\ell \in [n]$ satisfies $\tau_{\calC}(\ell) = \Omega(\eps)$ is $\Omega(\eps/\log n)$. It suffices to pick $\Theta(\log n / \eps)$ uniformly random $\ell \in [n]$ in order for one to satisfy $\tau_{\calC}(\ell) = \Omega(\eps)$ with high probability; and, if this event holds, then there exists $w \in [\log n]$ for which $\tau_\calC(\ell, w) = \Omega(\eps/\log n)$. We now leverage the querying paradigm described in the previous paragraph:
		if for any $\ell \in [n]$ as above and any $w \in [\log n]$ we query $\approx \log n / \eps$ uniform locations in $[\ell-2^w, \ell+2^w]$, then we shall find a $(12)$-pattern with good probability. In total, this procedure makes $O(\log^3 n / \eps^2)$ non-adaptive queries.

		To deal with general fixed $k \geq 2$ and $\eps > 0$, the (essentially) same reasoning is applied recursively, leading to the $(\log n)^{O(k^2)}$-query algorithm of \cite{NRRS17}. 

	\paragraph{Structural decomposition.}
		\cite{BCLW19} established a structural theorem for functions $f\colon [n] \to \R$ that are $\eps$-far from $(12\dots k)$-free, which led to improved non-adaptive algorithms. Specifically, they show that any $f$ which is $\eps$-far from $(12\dots k)$-free satisfies at least one of two conditions: either $f$ contains many \emph{growing suffixes}, or it can be decomposed into \emph{splittable intervals}. For the purpose of this discussion, let $\calC$ be any collection of $\Theta_{k, \eps}(n)$ disjoint $(12\dots k)$-copies in $f$.\footnote{To simplify the discussion, in the rest of this exposition we will generally not be interested in the exact dependence on the parameters $k$ and $\eps$, and for convenience we often use notions like $O_{k,\eps}(\cdot)$ and $\Omega_{k,\eps}(\cdot)$ that hide this dependence.}

		\begin{itemize}
			\item \textbf{Growing suffixes:} 
				there exist $\Omega_{k,\eps}(n)$ values of $\ell \in [n]$ where\footnote{We have previously defined the notions of cutting with slack and density only for the case $k = 2$, but they generalize rather naturally to any $k$. First, define the \emph{gap index} of a $(12 \dots k)$-pattern in $f$ in locations $x_1 < \ldots < x_k \in [n]$ as the smallest integer $c \in [k-1]$ maximizing $x_{c+1} - x_c$; the above copy is cut by $\ell$ with slack if $x_i + (x_{c+1}-x_c)/3 \leq \ell \leq x_{c+1} - (x_{c+1}-x_c)/3$. The gap-width of the copy is $\log(x_{c+1} - x_c)$. The definitions of $\calC_{\ell, w}$, $\tau_{\calC}(\ell, w)$, $\tau_{\calC}(\ell)$ can then be generalized in a straightforward manner, replacing ``width'' with ``gap width'' wherever relevant.} $\tau_{\calC}(\ell) \geq \Theta_k(\eps)$ and $\tau_{\calC}(\ell, w) \le \Theta(\tau_{\calC}(\ell) / k)$ for every $w \in [\log n]$. In other words, many $\ell \in [n]$ have that the sum of local densities, $\tau_{\calC}(\ell)$ of $(12\dots k)$-patterns in intervals of growing widths is not too small, and furthermore, the densities are not concentrated on any small set of widths $w$. Any such $\ell$ is said to be the starting point of a growing suffix.
			\item \textbf{Splittable intervals:} 
				there exist $c \in [k-1]$ and a collection of pairwise-disjoint intervals $I_1, \ldots, I_s \subset [n]$ with $\sum_{i=1}^{s} |I_i| = \Theta_{k, \eps}(n)$, so that each $I_i$ contains a dense collection of disjoint $(12 \dots k)$-patterns of a particular structure. Specifically, each such interval $I_i$ can be partitioned into three disjoint intervals $L_i, M_i, R_i$ (in this order), each of size $\Omega_k(|I_i|)$, where $I_i$ fully contains $\Omega_{k, \eps}(|I_i|)$ disjoint copies of $(12\dots k)$-patterns, in which the first $c$ entries lie in $L_i$, and the last $k-c$ entries lie in $R_i$ (none of these entries lies in $M_i$).
%
		\end{itemize}

	\ignore{\color{red}

		The high-level idea is as follows. First, a greedy rematching procedure from \cite{BCLW19} shows that any $f$ that is $\eps$-far from $(12\dots k)$-free must contain a collection $\calC$ of $(12\dots k)$-patterns with $|\calC| = \Theta_{k, \eps}(n)$, where all copies in $\calC$ have the same gap index $c \in [k-1]$, and $\calC$ also satisfies the requirement that if two copies in locations $x_1 < \ldots < x_k$ and $x'_1 < \ldots < x'_k$ satisfy $x_c < x'_c$ but $x_{c+1} > x'_{c+1}$, then $f(x_{c+1}) > f(x'_{c+1})$. Fix such a choice of $\calC$ for the rest of the discussion. 

		Now, what happens if some $\ell \in [n]$ is the start of the growing suffix, and for any $w \in [\log n]$ we query $f$ in $\Theta_{k, \eps}(1)$ uniformly random entries from the interval $[\ell, \ell+2^w]$? A simple quantitative analysis (which we omit here) based on the growing suffixes condition shows that with good probability, there will be a collection of $(12\dots k)$-copies 
		$x_1, \ldots, x_k$, where $x_i = (x_{i,1}, \ldots, x_{i,k}) \in \calC(\ell, w_i)$ and $w_{i+1} - w_i \ge 10$ (say) for every $i \in [k-1]$, such that the algorithm queries the $(c+1)$-entry -- $x_{i,c+1}$ -- of each $x_i$.
		Since $\ell$ cuts all the above copies with slack, it must hold that
		$$
			x_{k,c} < \ldots < x_{1,c} 
			< \ell 
			< x_{1,c+1} < \ldots < x_{k, c+1}.
		$$
		In view of our requirement on $\calC$, it immediately follows that $f(x_{k,c+1})  > \ldots > f(x_{1,c+1})$. In particular, we have detected a $(12\dots k)$-copy $(x_{1,c+1}, \ldots, c_{k,c+1})$, as desired.
		To summarize, if we pick $\Theta_{k,\eps}(1)$ uniformly random choices of $\ell \in [n]$, and apply the above querying procedure for each such $\ell$, a $(12 \dots k)$-copy will be detected with good probability assuming the growing suffixes condition holds. The total (non-adaptive) query complexity is $\Theta_{k,\eps}(\log n)$ as claimed.

		\comment{Shoham: I think this notation is cleaner, if not ideal}

		When the growing suffixes condition does not hold, \cite{BCLW19} recursively applies the splittable interval condition multiple times to get a non-adaptive algorithm making $O((\log n)^{\floor{\log_2 k}})$ queries. This cannot be improved, due to the matching $\Omega((\log n)^{\floor{\log_2 k}})$ lower bound proved there. Thus, in order to achieve a query complexity of $O_{k,\eps}(\log n)$ assuming the splittable intervals condition, our algorithm must be adaptive.}
	
		\cite{BCLW19} proceeds by devising an $O_{k,\eps}(\log n)$-query non-adaptive algorithm for the growing suffixes case, and an $O_{k,\eps}((\log n)^{\floor{\log_2 k}})$-query non-adaptive algorithm for the splittable intervals case. Thus, in order to obtain an $O_{k, \eps}(\log n)$-query \emph{adaptive} algorithm, it suffices to develop such an algorithm under the splittable intervals assumption.
	
	\paragraph{Robustifying the structural decomposition.} 

		\ignore{ \color{red}
			We now hope to devise an $O_{k, \eps}(\log n)$-queries adaptive algorithm for finding a $(12 \dots k)$-copy, assuming that $f$ satisfies the splittable intervals condition. Perhaps the most natural approach for this in view of the splittable intervals condition is as follows. (i) Approximate, in some way, the endpoints of some splittable interval $I_j$ and its left and right part $L_j, R_j$, while enumerating over the gap index $c \in [k-1]$; (ii) Make a recursive call searching for a $(12\dots c)$-copy $(x_1, \ldots, x_c)$ contained in the left part $L_i$ and a $(12 \dots k-c)$-copy $(y_1, \ldots, y_{k-c})$ in the right part $R_i$, with the hope that they will combine together to a $(12 \dots k)$-copy (for this to happen, we also need them to be compatible, in the sense that $f(x_c) < f(y_1)$). 

			In order to carry on with the approach suggested above, one can try to locate a ``$1$-entry'' lying in the left part $L_j$ and a compatible ``$(c+1)$-entry'' lying in the right part $R_j$ (here we henceforth we fix $c \in [k-1]$ and assume the gap of $I_j$ to equal $c$). Since in total $\Omega_{k,\eps}(n)$ entries $\ell \in [n]$ serve as the 1-entry lying in the left part $L_j$ of some interval $I_j$ where $j \in [s]$, it takes only a constant number of (non-adaptive) random queries from $[n]$ in order to hit one such $1$-entry $x \in [n]$.
			Suppose, then, that such a (sufficiently well-behaved) ``1-entry'' $x \in L_j$ was hit, and that the length of the containing splittable interval $|I_j|$ is unknown to us. The task now is to approximate this length, or (roughly) equivalently, to find $(c+1)$-entries compatible with $x$, which lie in the right part, $R_j$.

			Inspired by the previous approaches, we can try to uniformly sample elements to the right of $x$ at all possible scales, that is, to sample $O_{k, \eps}(1)$ such elements in $[x+2^{w-1}, x+2^w]$, for any possible $w \in [\log n]$.
			Among the queried elements, only those elements $y$ satisfying $f(y) > f(x)$ can serve as candidates to be $(c+1)$-entries in $R_j$, and it can be shown that if indeed $x$ is a (well-behaved) $1$-entry of some copy in $I_j$, then ``true positives'' -- well-behaved $(c+1)$-entries $y \in R_j$ satisfying $f(y) > f(x)$ -- will, indeed, be queried by this procedure. However, we might \emph{overshoot} and see many ``false positives'' among the queries: elements $y' \notin R_j$ that satisfy $f(y') > f(x)$, yet do not belong to the interval $I_j$, and in fact satisfy that $y'-x$ is much bigger than $|I_j|$. Without the ability to deal with overshooting, or with distinguishing true positives from false ones, it is unclear how to determine, or even approximate, the length of $I_j$, and we are seemingly stuck. }
		

		The splittable intervals condition, however, does not seem strong enough for our purposes: in order to utilize it, one would seemingly have to ``identify'', in some way, which parts of our sequence constitute splittable intervals, which is not clear how to do efficiently.	In order to bypass this issue, we substantially strengthen the structural theorem. The stronger statement asserts that any $f \colon [n] \to \R$ that is $\eps$-far from $(12 \dots k)$-free either satisfies the growing suffixes condition, defined previously, or a \emph{robust} version of the splittable intervals condition, defined as follows.
		\begin{itemize}
			\item \textbf{Robust splittable intervals:} 
				there exist $c \in [k-1]$ and a collection of pairwise-disjoint intervals $I_1, \ldots, I_s \subset [n]$ satisfying the same properties as in the ``splittable intervals'' setting described above (with slightly different dependence on $\eps$ and $k$ in the $\Theta_{k, \eps}(\cdot)$ term). Additionally, \emph{any} interval $J \subset [n]$ which contains an interval $I_j$ is itself far from $(12 \dots k)$-free, i.e.\ it contains a collection of $\Omega_{k,\eps}(|J|)$ disjoint $(12\dots k)$-copies.
		\end{itemize}
		
		\ignore{The (stronger) robust splittable intervals condition follows from the original condition, of (non-robust) splittable intervals; the proof is strikingly simple, relying on an elementary counting argument. Let $I_1, \ldots, I_s$ be the splittable intervals in the non-robust setting. Let $\calI$ be the collection of all intervals $I_j$ where there exists an interval $J_j$ that does not contain $c_{k,\eps}|J_j|$ disjoint $(12\dots k)$-copies (for a sufficiently small $c_{k,\eps} > 0$). Finally, let $\calJ$ be a minimal set of intervals $\{J_j : I_j \in \calI\}$. From the minimality, it can be shown that any $x \in [n]$ is contained in at most three intervals from $\calJ$, and so $\sum_{J \in \calJ} |J| \leq 3n$. Thus, the total number of disjoint $(12\dots k)$-copies in intervals from $\calJ$ (and so, in intervals from $\calI$) is bounded by $3n \cdot c_{k, \eps}$, meaning in particular that the intervals in the collection $\calI' = \{I_1, \ldots, I_s\} \setminus \calI$ contain $\Theta_{k, \eps}(n)$ disjoint $(12\dots k)$-copies in total provided that $c_{k,\eps}$ is small enough. One can now verify that $\calI'$ satisfies all requirements of the robust splittable intervals condition. }

	\paragraph{Towards an algorithm.} 
		At a high level, the algorithms of \cite{BCLW19} and \cite{NRRS17} proceed in a recursive manner where each step tries to find the relevant width considered (which is one of $\Omega(\log n)$ options). Since their algorithms are non-adaptive, they consider all $\Omega(\log n)$ options in recursive steps, and hence, suffer a logarithmic factor with each step. Since our algorithm is adaptive, we want to choose \emph{one} of the widths to recurse on. The algorithm will ensure that the width considered is large enough. When the width chosen is not too much larger, our recursive step proceeds similarly to \cite{NRRS17}; we call this the \emph{fitting case}. However, the width considered may be too large; we call this case \emph{overshooting}. In order to deal with the overshooting case, we algorithmically utilize the robust structural theorem in a somewhat surprising manner in order to detect a $(12\dots k)$-copy.

		We now expand on the above idea and provide an informal description. As \cite{BCLW19} gives an $O_{k, \eps}(\log n)$-query algorithm when our function $f$ satisfies the growing suffixes condition, we may assume that $f$ satisfies the \emph{robust splittable intervals} condition. Consider sampling, for $O_{k, \eps}(1)$ repetitions, an index $\bx \in [n]$ uniformly at random, and for each $t \in [\log n]$, a random index $\by_t \in [\bx, \bx + 2^t]$. Consider the following event: 
		\begin{quote}
			The index $\bx$ is a (sufficiently well-behaved)\footnote{Recall that, in the first polylogarithmic-query qlgorithm described above, we hoped to hit a ``1-entry'' $x$ whose value $f(x)$ is no higher than some suitable median value; the ``well-behaved'' requirements are of similar flavor, and do not incur more than a constant overhead on the query complexity.} first element in some $(12\dots k)$-pattern falling in some robust splittable interval $I_j$, and for $t^* \in [\log n]$ satisfying $|I_j| \leq 2^{t^*} \leq 2|I_j|$, $\by_{t^*}$ is a (well-behaved) $(c+1)$-th element in some $(12\dots k)$-pattern falling in $I_j$.
		\end{quote}
		We claim that the above event occurs with high (constant) probability for at least one choice of $\bx$, and that when this event does occur, the algorithm can be recursively applied without incurring a multiplicative logarithmic factor. Indeed, suppose that the above holds for some $\bx$.\footnote{More precisely, our algorithm runs this procedure for any of our choices of $\bx$, without ``knowing'' which of them satisfies the above event. Since the total number of choices is $O_{k, \eps}(1)$, this incurs only a constant overhead.} We set $\by$ to be $\by_t$, where $t$ is the largest such that $f(\bx) < f(\by_t)$ holds, and notice in particular that $t \geq t^*$. This means that $\bx < \by$ and $f(\bx) < f(\by)$.

		The \emph{fitting case} occurs when $t$ (achieving the maximum above) is roughly the same as $t^*$. To handle this case, we recurse by finding a $(12\dots c)$-patterns in $L_j$, and $(12\dots (k-c))$-pattern in $R_j$. At a high level, if one takes $\Theta_{k,\eps}(1)$ independent uniform samples $\bz$ from $[\bx, \by]$, then one of them is likely to fall in the middle part $M_j$ of $I_j$, so that $L_j \subset [\bx - 2^t, \bz]$ and $R_j \subset [\bz, \by + 2^t]$, allowing us to proceed recursively. While this description omits a few details, the intuition proceeds similarly to \cite{NRRS17}, except that the recursion occurs only on one width, namely, $t$, and does not lose multiplicative logarithmic factors as in the previous approaches. 

	\paragraph{The overshooting component.}
		The other case, of \emph{overshooting}, occurs when $t$ is significantly larger than $t^*$. We expand on the main ideas here in more detail; the strong guarantee given by the robust splittable intervals condition adds a ``for all'' element into the structural characterization, which is able to treat the problem posed by overshooting in a rather surprising and non-standard way. Since $t$ is much larger than $\log|I_j|$, there exist $k-2$ intervals $J_1, \ldots, J_{k-2} \subset [\bx,\by]$ satisfying the following conditions:
		\begin{itemize}
			\item 
				$J_1$ lies immediately after the interval $I_j$ (recall that $I_j$ is the interval containing $\bx$).
			\item 
				$J_{i+1}$ lies immediately after $J_i$, for any $i \in [k-3]$.
			\item 
				$|J_1| = \ceil{|I_j| \cdot \alpha_{k,\eps}}$ and $|J_{i+1}| = \ceil{|J_i| \cdot \alpha_{k,\eps}}$ for any $i \in [k-3]$, for some large enough $\alpha_{k,\eps} > 1$.
		\end{itemize}
		For any $i \in [k-2]$, set $J'_i$ to be the minimal interval containing both $I_j$ and $J_i$. The robust splittable intervals condition asserts that (since each $J'_i$ contains the splittable interval $I_j$) the number of disjoint $(12\dots k)$-copies in $J'_i$ is proportional to $|J'_i|$, and provided that $\alpha_{k, \eps}$ is large enough, this means that $J_i = J'_i \setminus J'_{i-1}$ also contains a collection $\calT_i$ of $\Omega_{k,\eps}(|J_i|)$ disjoint $(12\dots k)$-copies.
		We now define two sets $\calA_i$ and $\calB_i$ as follows. 
		Let $\calA_i$ be the collection of prefixes $(a_1, \ldots, a_{i+1})$ of $k$-tuples from $\calT_i$ with $f(a_{i+1}) < f(y)$, and let $\calB_i$ be the collection of suffixes $(a_{i+1}, \ldots, a_k)$ of $k$-tuples from $\calT_i$ with $f(a_{i+1}) \ge f(y)$. As $|\calT_i| = |\calA_i| + |\calB_i|$, one of $\calA_i$ and $\calB_i$ is large (i.e.\ has size at least $\Omega_{k, \eps}(|J_i|)$).
		%
		%

		This seemingly innocent combinatorial idea can be exploited non-trivially to find a $(12\dots k)$-copy.
		Specifically, the algorithm to handle overshooting aims to find (recursively) shorter increasing subsequences in $J_1, \ldots, J_{k-2}$, with the hope of combining them together into a $(12\dots k)$-copy. 
		Concretely, for any $i \in [k-2]$, 
		we make two recursive calls of our algorithm on $J_i$: one for a $(k-i)$-increasing subsequence in $J_i$ whose values are at least $f(y),$\footnote{Technically speaking, our algorithm can be configured to only look for increasing subsequences whose values lie in some range; we use this to make sure that shorter increasing subsequences obtained from the recursive calls of the algorithm can eventually be concatenated into a valid length-$k$ one.}
		and a second for an $(i+1)$-increasing subsequence in $J_i$, with values smaller than $f(y)$.
		By induction, the first recursive call succeeds with good probability if $|\calA_i|$ is large, while the second call succeeds with good probability if $|\calB_i|$ is large. Since for any $i$ either $|\calA_i|$ or $|\calB_i|$ must be large, at least one of the following must hold.
		\begin{itemize}
			\item 
				$\calB_1$ is large. 
				In this case we are likely to find a length-$(k-1)$ monotone pattern in $J_1$ with values at least $f(y) > f(x)$, which combines with $\bx$ to form a length-$k$ monotone pattern.
			\item 
				$\calA_{k-2}$ is large. 
				Here we are likely to find a length-$(k-1)$ monotone pattern in $J_{k-2}$ whose values lie below $f(y)$, which combines with $\by$ to form a length-$k$ monotone pattern.
			\item 
				There exists $i \in [k-3]$ where both $\calA_i$ and $\calB_{i+1}$ are large. 
				Here we will find, with good probability, a length-$(i+1)$ monotone pattern in $J_i$ with values below $f(y)$, and a length-$(k-i-1)$ monotone pattern in $J_{i+1}$ with values above $f(y)$; together these two patterns combine to form a $(12 \ldots k)$-pattern.
		\end{itemize}

		In all cases,  a $k$-increasing subsequence is found with good probability. 

		Finally, for the query complexity, our algorithm (which runs both the ``fitting'' component and the ``overshooting'' component, to address both cases) makes $O_{k,\eps}(\log n)$ queries: each call makes $O_{k,\eps}(\log n)$ queries in itself and $O_{k,\eps}(1)$ additional calls recursively, where the recursion depth is bounded by $k$. It follows that the total query complexity is of the form $O_{k,\eps}(\log n)$.

\ignore{In this case, we will use the robust structural result to identify a collection of $k-2$ adjacent intervals $J_1, \dots, J_{k-2}$ which lie between $[\bx, \by]$ and each contains many $(12\dots k)$-patterns. We recurse twice on each $J_{\kappa}$ for $\kappa \in [k-2]$: in one execution, we seek a $(12\dots (\kappa+1))$-pattern lying \emph{below} $f(\by)$, and in the other execution, we seek a $(12\dots (k-\kappa))$-pattern lying \emph{above} $f(\by)$. Since each $J_{\kappa}$ contains many $(12\dots k)$-patterns, one of the two executions will succeed. Lastly, these intervals lie within $[\bx, \by]$ and $f(\bx) < f(\by)$; which will mean that we will be able to combine some of the patterns found to form a $(12\dots k)$-pattern.}
	
\ignore{	Recall the above proposed algorithm to approximate an interval $I_j$ by hitting a $1$-entry from $L_j$ and $(c+1)$-entry from $R_j$. Let $x$ be a candidate for a $1$-entry, which we assume to lie in some splittable interval $I_j$, whose characteristics are unknown to us at this point; the probability of a random $x \in [n]$ to indeed be a valid $1$-entry in some interval is $\Omega_{k, \eps}(1)$, so we can take $O_{k,\eps}(1)$ choics of $x$ and run what follows for each $x$ separately. For each $w \in [\log n]$, make $\Theta_{k,\eps}(1)$ random queries in $[x, x + 2^{w-1}]$, and take $y$ to be the rightmost (in terms of location) among the queried elements $z$ for which $f(z) > f(x)$. 
	Assuming that $x$ is indeed a valid $1$-entry, with good probability we query an element $z$ which is a $(c+1)$-entry in $R_j$. Assuming that such an element is indeed queried, our algorithm treats separately the \emph{overshooting} case -- the case where $y$ is not contained in the interval $I_j$ -- and the \emph{fitting} case, in which $y$ is in the interval (or in close proximity to it). Note that our algorithm does not ``know'' at this point which of the cases holds, if at all.

\paragraph{Handling the overshooting case.}
	The strong guarantee given by the robust splittable intervals conditions adds a ``for all'' element into the structural characterization, which is able to treat the problem posed by overshooting in a rather surprising and non-standard way. Given $x$ and $y$ as in the last paragraph, if $\log(y-x) - \log(|I_j|)$ is large enough as a function of $k,\eps$, then there exist $k$ intervals $J_1, \ldots, J_{k-2} \subset [x,y]$ satisfying the following conditions:
	\begin{itemize}
		\item 
			$J_1$ lies immediately after the interval $I_j$ (which is the interval containing $x$).
		\item 
			$J_{i+1}$ lies immediately after $J_i$, for any $i \in [k-3]$.
		\item 
			$|J_{i+1}| = \ceil{|J_i| \cdot \alpha_{k,\eps}}$ for any $i \in [k-3]$, where again $\alpha_{k,\eps} > 1$ is large enough as a function of $k$ and $\eps$. Moreover, $|J_1| = \ceil{|I_j| \cdot \alpha_{k,\eps}}$.
	\end{itemize}
	For any $i \in [k-2]$, set $J'_i$ to be the minimal interval containing both $I_j$ and $J_i$. The robust splittable intervals condition asserts that (since each $J'_i$ contains the splittable interval $I_j$) the number of disjoint $(12\dots k)$-copies in $J'_i$ is proportional to $|J'_i|$, and provided that $\alpha_{k, \eps}$ is large enough, this means that $J_i = J'_i \setminus J'_{i-1}$ also contains a collection $\calT_i$ of $\Theta_{k,\eps}(|J'_i|) = \Theta_{k,\eps}(|J_i|)$ disjoint $(12\dots k)$-copies.
	We now define two sets $\calA_i$ and $\calB_i$ as follows. Initialize both sets to be empty, and for any copy $(a_1, \ldots, a_k) \in \calT$, add $(a_{i+1}, \ldots, a_k)$ to $\calA_i$ if $f(a_{i+1}) \geq f(y)$, otherwise add $(a_1, \ldots, a_{i+1})$ to $\calB_i$. In the end of the process, either $\calA_i$ contains $\Theta_{k,\eps}(|J_i|)$ disjoint $(12\dots k-i)$ copies whose smallest element is at least than $f(y)$, or $\calB_i$ contains $\Theta_{k,\eps}(|J_i|)$ disjoint $(12\dots i+1)$ copies whose largest element is smaller than $f(y)$. 
	%
	%

	This seemingly innocent combinatorial idea can be exploited non-trivially to find a $(12\dots k)$-copy.
	Specifically, the algorithm to handle overshooting aims to find (recursively) shorter increasing subsequences in $J_1, \ldots, J_{k-2}$, with the hope of combining them together into a $(12\dots k)$-copy.\footnote{Technically speaking, to make sure that shorter increasing subsequences can be combined into a longer one, our (recursive calls to the) algorithm receives, as part of its input parameters, the range of values in which it is required to find an increasing subsequence, as well as the interval in which the subsequence should reside. For example, if we require one call of the algorithm to have values in $(-\infty, a)$ for some $a \in \R$ and locations in $[1, n_1]$, while another call to the algorithm receives the value range $[a,b]$ for some $b \geq a$ and the interval $[n_1+1, n]$ as an input,  then we can rest assured that outputs from these two calls can be combined into a longer increasing subsequence.} Concretely, for any $i \in [k-2]$, 
	we make two recursive calls of our algorithm on $J_i$:
	\begin{enumerate}
		\item 
			The first recursive search is for a $(k-i)$-increasing subsequence in the interval $J_i$, whose values are at least $f(y)$.\footnote{More accurately, the range of allowed values for the recursive call is the intersection of the input range with $[f(y), \infty)$, and the interval in which this call should operate is $J_i$.}
		\item 
			The second search is for an $(i+1)$-increasing subsequence in $J_i$, with values smaller than $f(y)$.
	\end{enumerate} 
	By induction, the first recursive call succeeds with good probability if $|\calA_i|$ is large (of order $\Theta_{k, \eps}(|J_i|)$), while the second call succeeds with good probability if $|\calB_i|$ is large. Since for any $i$ either $|\calA_i|$ or $|\calB_i|$ must be large, at least one of the following must hold.
	\begin{itemize}
		\item 
			$|\calA_1|$ is large. In this case, with good probability the first call on $J_1$ returns an increasing subsequence $z_1 < \ldots < z_{k-1}$ where $f(z_1) \geq f(y) > f(x)$. By concatenating $x$ with this subsequence, a $k$-increasing subsequence is formed, as desired.
		\item 
			$|\calB_{k-2}|$ is large. In this case, with good probability the second call on $J_{k-2}$ will return a $(k-1)$-increasing subsequence with all values smaller than $f(y)$, which combined with $y$ forms a $(12 \dots k)$-copy.
		\item 
			Otherwise, there exists some $i \in [k-3]$ where both $|\calB_i|$ and $|\calA_{i+1}|$ are large. Here, with good probability an $(i+1)$-increasing subsequence from $J_i$ with maximum value less than $f(y)$ and a $(k-i-1)$-increasing subsequence from $J_{i+1}$ with minimum value at least $f(y)$ will be found, and together they can be concatenated to form a $k$-increasing subsequence as desired.
	\end{itemize}

	In all cases, a $k$-increasing subsequence is found with good probability. This rather surprising technique settles the overshooting case, where $y$ is far from the interval $I_j$ in which $x$ lies. We now turn to the other regime, where a more standard approach suffices.

\paragraph{Handling the fitting case.} 
	Now suppose that $x$ and $y$ are as above, and suppose that $\log(y-x) - \log|I_j| = O_{k, \eps}(1)$.
	In this case, $x$ and $y$ serve as relatively good estimations for the endpoints of $I_j$. This implies that there is a family $\calC_j$ of $\Omega_{k, \eps}(|I_j|)$ disjoint copies of $(12\ldots k)$ such that the $c$-entry of any $(12\ldots k)$-copy in $\calC_j$ lies to the left and below each $(c+1)$-entry of an $(12 \ldots k)$-copy in $\calC_j$.
	
	We make $O_{k,\eps}(1)$ random queries in order to find, with good probability, an element $z \in [x,y]$ such that for at least $\Omega_{k,\eps}(|I_j|)$ of the subsequences in $\calC_j$, their $(c+1)$-entry lies to the right of $z$ while their $c$-entry lies to the left of $z$ (note that any $z \in M_j$ suffices for this purpose). Next, we make $O_{k,\eps}(1)$ random queries in order to find, with good probability, an element $w \in [z,y]$ such that $f(w)$ lies below the $(c+1)$-entry of at least $\Omega_{k,\eps}(|I_j|)$ subsequences in $\calC_j$ whose $(c+1)$-entries are also to the right of $z$, and above the $c$-entry of at least $\Omega_{k,\eps}(|I_j|)$ subsequences in $\calC_j$ whose $c$-entry is to the left of $z$. Putting everything together, we find that there exists a collection of $\Omega_{k,\eps}(|I_j|)$ disjoint $c$-increasing subsequences in $[x, z)$ whose elements lie (strictly) below $w$, and there is a collection of $\Omega_{k,\eps}(|I_j|)$ disjoint $(k-c)$-increasing subsequences in $[z, y]$ whose elements lie above $w$. Two recursive call of the algorithm will find, with high probability, a $c$-increasing subsequence in $[x,z)$ that lies below $w$, and a $(k-c)$-increasing subsequence in $[z, y]$ that lies above $w$, which toghether form the required $k$-increasing subsequence.

	\comment{Shoham: updated the description here. I noticed that Erik (as far as I can tell) uses only one element $w$ to determine how the interval $[x,y]$ is split and also how the interval of possible values is split. I believe that's actually a mistake: if the $(c+1)$-entries form a decresasing subsequence, that an appropriate $w$ does not exist. 

	Anyway, I put this description here which I think should work as written (in my old write-up I also `guessed' $|I_j|$ but that's probably unnecessary), and will look into the details in Section 4 to make sure everything there works.}

\paragraph{Putting it all together.}
	To summarize, our adaptive $O_{k,\eps}(\log n)$-query algorithm to find a $k$-increasing subsequence in $f \colon [n] \to \R$, assuming $f$ is $\eps$-far from free of such subsequences, goes as follows.
	\begin{itemize}
		\item 
			First, we try to apply the (non-adaptive) algorithm for finding growing suffixes. If the growing suffixes condition holds, then this step will find the desired increasing subsequence with good probability.
		\item 
			Otherwise, the robust splittable intervals condition holds. We pick $O_{k,\eps}(1)$ candidates $x$ for $1$-entries, and, for each such $x$ and each $w \in [\log n]$, we sample $O_{k,\eps}(1)$ random elements in $[x+2^w, x+2^{w+1}]$.
		\item 
			Let $y$ be the rightmost element queried for which $f(y) > f(x)$. We run the overshooting algorithm. If indeed $x$ is a well-behaved $1$-entry of some splittable interval $I_j$ and $y$ is an overshoot with respect to $x$ and $I_j$, then this step will find a $k$-increasing subsequence with good probability. 
		\item 
			Otherwise, $y$ is not an overshoot with respect to $x$, and we employ the algorithm for the fitting case.
	\end{itemize}
	To analyze the query complexity, note that the only point in the algorithm where the query complexity depends on $n$ is in the second step when elements are sampled at all possible scales with respect to $x$. It is not hard to verify that the query complexity is thus of the form $T(k, \eps, n) = O_{k,\eps}(\log n)$. Indeed, it is bounded by an expression of the form 
	$$O_{k,\epsilon}(\log n) + \sum_{k', \eps'} T(k', \eps', n),$$
	where $(k', \eps')$ ranges over all tuples of parameters for the recursive calls made by our algorithm (these parameters depend only on $k,\eps$ and not on $n$); in all such tuples, $k' < k$, which gives the desired bound.

	\comment{Where should we talk about the lower bound (reduction from monotonicity testing)? I think it was first shown by Newman, and is a very simple construction nonetheless. Maybe just give the statement and mention that it follows from Newman's paradigm.}}


\subsection{Notation}

All logarithms considered are base $2$. We consider functions $f\colon I \to \R$, where $I \subseteq [n]$, as the inputs and main objects of study. An \emph{interval} in $[n]$ is a set $I \subseteq [n]$ of the form $I = \{a, a+1, \ldots, b\}$. At many places throughout the paper, we think of augmenting the image with a special character $*$ to consider $f \colon I \to \R \cup \{ *\}$. 
$*$ can be thought of as a \emph{masking}  operation: In many cases, we will only be interested in entries $x$ of $f$ so that $f(x)$ lies in some prescribed (known in advance) range of values $R \subseteq \R$, so that entries outside this range will be marked by $*$.
Whenever the algorithm queries $f(x)$ and observes $*$, it will interpret this as an incomparable value (with respect to ordering) in $\R$. As a result, $*$-values will never be part of monotone subsequences. We note that augmenting the image with $*$ was unnecessary in \cite{NRRS17, BCLW19} because they only considered non-adaptive algorithms. We say that for a fixed $f\colon I \to \R \cup \{*\}$, the set $T$ is a collection of disjoint monotone subsequences of length $k$ if it consists of tuples  $(i_1, \dots, i_k) \in I^k$, where $i_1 < \dots < i_k$ and $f(i_1) < \dots < f(i_k)$, and furthermore, for any two tuples $(i_1, \ldots, i_k)$ and $(i'_1, \ldots, i'_k)$, their intersection (as sets) is empty. We also denote $E(T)$ as the union of indices in $k$-tuples of $T$, i.e., $E(T) = \cup_{(i_1,\dots, i_k) \in T} \{ i_1, \dots, i_k\}$. Finally, we let $\poly( \cdot )$ denote a large enough polynomial whose degree is (bounded by) a universal constant. 

\section{Stronger structural dichotomy}
	In this section, we establish the structural foundations -- specifically, the \emph{growing suffixes} versus \emph{robust splittable intervals} dichotomy -- lying at the heart of our adaptive algorithm.
	We start with the definitions. The first is the definition of a growing suffix setting, as given in \cite{BCLW19}. For what follows, for an index $\ell \in [n]$ define $\eta_\ell = \ceil{\log_2(n-\ell)}$, and for any $t \in [\eta_\ell]$ set $S_t(\ell) = [a+2^{t-1}, a+2^t) \cap [n]$. Note that the intervals $S_1, \ldots, S_{\eta_\ell}$ are a partition of $(\ell, n]$ into intervals of exponentially increasing length (except for maybe the last one). Finally, the tuple $S(\ell) = (S_t(\ell))_{t \in [\eta_\ell]}$ is called the \emph{growing suffix} starting at $\ell$. 


	\begin{definition}[Growing suffixes (see \cite{BCLW19}, Definition 2.4)]\label{def:growing-suffixes}
		Let $\alpha, \beta \in [0,1]$.
		We say that an index $\ell \in [n]$ starts an \emph{$(\alpha, \beta)$-growing suffix} if, when considering the collection of intervals $S(\ell) = \{ S_t(\ell) : t \in [\eta_\ell]\}$, for each $t \in [\eta_\ell]$ there is a subset $D_t(\ell) \subseteq S_t(\ell)$ of indices 
		such that the following properties hold. 
		\begin{enumerate} 
			\item 
				We have $|D_t(\ell)|/|S_t(\ell)| \le \alpha$ for all $t \in [\eta_\ell]$, and
				$\sum_{t=1}^{\eta_\ell} |D_t(\ell)|/|S_t(\ell)| \geq \beta$.
				\label{en:grow-cond-3}
			\item 
				For every $t, t' \in [\eta_a]$ where $t < t'$, if $a \in D_t(\ell)$ and $a' \in D_{t'}(\ell)$, then $f(a) < f(a')$. 
				\label{en:grow-cond-2}
		\end{enumerate}
	\end{definition}
	
	The second definition, also from \cite{BCLW19}, describes the (non-robust) splittable intervals setting. 
	
	\begin{definition}[Splittable intervals (see \cite{BCLW19}, Definition 2.5)]\label{def:splittable}
		Let $\alpha,\beta \in (0,1]$ and $c \in [k-1]$. Let $I \subseteq [n]$ be an interval, let $T \subseteq I^{k}$ be a set of disjoint, length-$k$ monotone subsequences of $f$ lying in $I$, and define
		\begin{align*} 
			T^{(L)} &= \{ (i_1, \dots, i_c) \in I^c : (i_1, \dots, i_c) \text{ is a prefix of a $k$-tuple in $T$}\}, \text{ and }\\
			T^{(R)} &= \{ (j_1, \dots, j_{k-c}) \in I^{k-c} : (j_1, \dots, j_{k-c}) \text{ is a suffix of a $k$-tuple in $T$}\}.
		\end{align*}
		We say that the pair $(I, T)$ is \emph{$(c, \alpha,\beta)$-splittable} if $|T|/|I| \geq \beta$; $f(i_c) < f(j_1)$ for every $(i_1, \dots, i_c) \in T^{(L)}$ and $(j_1, \dots, j_{k-c}) \in T^{(R)}$; and there is a partition of $I$ into three adjacent intervals $L, M, R \subseteq I$ (that appear in this order, from left to right) of size at least $\alpha |I|$, satisfying $T^{(L)} \subseteq L^c$ and $T^{(R)} \subseteq R^{k-c}$. 
		
		A collection of disjoint interval-tuple pairs $(I_1, T_1), \dots, (I_s, T_s)$ is called a \emph{$(c, \alpha,\beta)$-splittable collection of $T$} if each $(I_j, T_j)$ is $(c, \alpha, \beta)$-splittable and the sets $(T_j : j \in [s])$ partition $T$.
	\end{definition}

	The following theorem presents the growing suffixes versus (non-robust) splittable intervals dichotomy, which is among the main structural results of \cite{BCLW19}.\footnote{In \cite{BCLW19}, the theorem is stated with respect to two parameters, $k, k_0$. For our purpose it suffices to set $k_0 = k$.} 
	\begin{theorem}[\cite{BCLW19}]\label{thm:main-structure}
		Let $k,n \in \N$, $\eps\in (0,1)$, and $C > 0$, and let $I \subseteq [n]$ be an interval. Let $f \colon I \to \R \cup \{\ast\}$ be a function and let $T^0 \subseteq I^{k}$ be a set of at least $\eps |I|$ disjoint monotone subsequences of $f$ of length $k$. 
		Then there exist $\alpha \in (0,1)$ and $p > 0$ satisfying $\alpha \geq \Omega(\eps/k^5)$ and $p \leq \poly(k \log(1/\eps))$ such that at least one of the following conditions holds.
		\begin{enumerate}
			\item \label{en:suffix} 
				\textbf{Growing suffixes:} There exists a set $H \subseteq [n]$, of indices that start an $(\alpha, C k \alpha)$-growing suffix, satisfying $\alpha |H| \geq (\eps/p)  n$.
			\item \label{en:split} 
				\textbf{Splittable intervals (non-robust):} There exist an integer $c$ with $1 \le c < k$, a set $T$, with $E(T) \subseteq E(T^0)$, of disjoint length-$k$ monotone subsequences, and a $(c, 1/(6k),\alpha)$-splittable collection of $T$, consisting of disjoint interval-tuple pairs $(I_1, T_1), \dots, (I_s, T_s)$, such that 
				\begin{align}
				\label{eqn:splittable_intervals_non_robust}
					\alpha \sum_{h=1}^s |I_h| \geq |T^0|/p. 
				\end{align}
		\end{enumerate}
	\end{theorem}

%

	As argued in Section \ref{subsec:techniques}, the splittable intervals condition does not seem strong enough by itself to be useful for adaptive algorithms. Therefore, we next aim to establish a stronger structural dichotomy, asserting that $f$ either satisfies the growing suffixes condition, or a \emph{robust} version of the splittable intervals condition. The next lemma will imply that the growing suffixes condition can be robustified by merely throwing away a subset of ``bad'' splittable intervals.

	\begin{lemma} \label{lem:intervals}
		Let $\alpha \in (0,1)$ and let $I \subset \N$ be an interval.
		Suppose that $I_1, \ldots, I_s \subset I$ are disjoint intervals such that $\sum_{h=1}^s |I_h| \ge \alpha |I|$.
		Then there exists a set $G \subset [s]$ such that 
		\[
			\sum_{h \in G} |I_h| \ge (\alpha/4)|I|,
		\]
		and for every interval $J \subset I$ that contains an interval $I_h$ with $h \in G$, 
		\[
			\sum_{h \in [s] \colon I_h \subset J} |I_h| \ge (\alpha/4)|J|.
		\]
	\end{lemma}


	\begin{proof}
		Let $B \subseteq [s]$ be the set of indices $h$ for which there is an interval $J_h \supseteq I_h$ satisfying $\sum_{h \in [s]:I_h \subseteq J} |I_h| < (\alpha / 4) |J|$. For each $h \in B$ fix such a containing interval $J(I_h)$.

		Let $\calJ$ be a minimal subset of $\{J(I_h) : h \in B\}$ with the following property: for any $h \in B$ there exists $J \in \calJ$ containing $I_h$. Such a minimal subset clearly exists, since $\{J(I_h) :h \in B\}$ itself satisfies this property (but is not necessarily minimal).
		The next claim asserts that no vertex is covered more than three times by sets in $\calJ$.
		\begin{claim}
			Every element $x \in I$ is contained in at most three intervals from $\calJ$.
		\end{claim}
		\begin{proof}
			The proof follows from the minimality of $\calJ$.
			Consider first the case where $x \in I_{h^*}$ for some $h^* \in B$. Let $J_{L} = [a_L, b_L]$ be an interval from $\calJ$ that contains $x$, and whose left-most element $a_L$ is furthest to the left among all intervals from $\calJ$ that contain $x$; pick $J_R = [a_R, b_R]$ symmetrically, with $b_R$ being furthest possible to the right; and let $J_M = [a_M, b_M]$ be an interval from $\calJ$ that contains $I_h$. We claim that $\calJ$ does not have any other intervals that contain $x$. Suppose, to the contrary, that there exists $J = [a, b] \in \calJ$ containing $x$ where $J \neq J_{L}, J_R, J_M$; note that by definition of $J_L$ and $J_M$, $a_L \le a$ and $b_R \ge b$.  

			We claim that $\calJ \setminus \{J\}$ covers all intervals $I_h$ with $h \in [B]$; it suffices to show that for any $h \in B$ such that $I_h \subset J$, one of the intervals $J_L, J_R, J_M$ covers $I_h$. Consider $h \in B$ such that $I_h \subset J$, and write $I_h = [c, d]$. If $h = h^*$, then $I_h \subset J_M$. If $I_h$ lies to the left of $I_{h^*}$, then $d < x \leq b_L$, and $c \ge a \ge a_L$, so $I_h \subseteq J_{L}$. Similarly, if $I_h$ lies to the right of $I_h$, then $I_h \subseteq J_R$. It follows that, indeed, intervals from $\calJ \setminus \{J\}$ cover all intervals in $\{I_h : h \in B\}$, contradicting the minimality of $\calJ$.
			
			Now, if $x$ is not contained in any interval of $I_h$ with $h \in B$, then we can show similarly that there are at most two intervals from $\calJ$ that contain $x$, by defining $J_{L}$ and $J_R$ as above. 
		\end{proof}
		Let $U$ be the union of intervals from $\calJ$. In light of the above claim,
		\[
			\sum_{h \in B} |I_h| 
			\le \sum_{J \in \calJ} \left(\sum_{h \in [s]: \, I_h \subseteq J} |I_h|\right)
			< \frac{\alpha}{4} \cdot \sum_{J \in \calJ} |J|
			\le \frac{3\alpha}{4} \cdot |U|
			\le \frac{3\alpha}{4} \cdot |I|,
		\]
		where the first inequality holds because each $I_h$ with $h \in B$ is covered by an interval in $\calJ$; the second inequality follows as $\calJ$ consists of sets $J(I_h)$ with $h \in B$; the third inequality follows from the claim; and the last one holds because $U \subset I$.
		Finally, let $G = [s] \setminus [B]$. By assumption on $\sum_h |I_h|$ and the previous line,
		\[
			\sum_{h \in G} |I_h| 
			= \sum_{h \in [s]} |I_h| - \sum_{h \in B} |I_h| 
			\geq \alpha |I| -  \frac{3\alpha}{4} \cdot |I| 
			= \frac{\alpha}{4} \cdot |I|,
		\]
		and every interval $J$ that contains an interval $I_h$ with $h \in G$ satisfies $\sum_{h \in [s] \colon I_h \subset J} |I_h| \ge (\alpha/4)|J|$, as required.
	\end{proof}

	The robust version of the structural dichotomy is stated below; the proof follows easily from the basic structural dichotomy in combination with the last lemma.
	\begin{theorem}[Robust structural theorem]\label{thm:main-structure-2}
		Let $k,n \in \N$, $\eps\in (0,1)$, and $C > 0$, and let $I \subseteq [n]$ be an interval. Let $f \colon I \to \R \cup \{\ast\}$ be an array and let $T^0 \subseteq I^{k}$ be a set of at least $\eps |I|$ disjoint length-$k$ monotone subsequences of $f$. 
		Then there exist $\alpha \in (0,1)$ and $p > 0$ with $\alpha \geq \Omega(\eps/k^5)$ and $p \leq \poly(k \log(1/\eps))$ such that at least one of the following holds.
		\begin{enumerate}
			\item 
				\textbf{Growing suffixes:} There exists a set $H \subseteq [n]$, of indices that start an $(\alpha, C k \alpha)$-growing suffix, satisfying $\alpha |H| \geq (\eps/p) n$.
			\item 
				\textbf{Robust splittable intervals:} There exist an integer $c$ with $1 \le c < k$, a set $T$, with $E(T) \subseteq E(T^0)$, of disjoint length-$k$ monotone subsequences, and a $(c, 1/(6k),\alpha)$-splittable collection of $T$, consisting of disjoint interval-tuple pairs $(I_1, T_1), \dots, (I_s, T_s)$, such that 
				\begin{align}
					\label{eqn:splittable-intervals-robust}
					\alpha \sum_{h=1}^s |I_h| \geq (\eps/p) |I|, 
				\end{align}
				Moreover, if $J \subset I$ is an interval where $J \supset I_h$ for some $h \in [s]$, $J$ contains at least $(\eps / p) |J|$ disjoint $(12\dots k)$-patterns from $T^0$.
		\end{enumerate}
	\end{theorem}

%

	\begin{proof}
		Apply Theorem~\ref{thm:main-structure}. Let $\alpha^* \in (0,1)$ and $p^*$ be parameters such that $\alpha^* \ge \Omega(\eps / k^5)$ and $p^* \le \poly(k\log(1/\eps))$, as guaranteed by the theorem. Set $\alpha = \alpha^*$ and $p = 4p^*$.
		If Condition \ref{en:suffix} holds in the application of Theorem~\ref{thm:main-structure}, then the analogous growing suffix condition in Theorem~\ref{thm:main-structure-2} clearly holds.
		So suppose that Condition~\ref{en:split} in Theorem~\ref{thm:main-structure} holds, and let $c$ and $(I_1, T_1), \ldots, (I_s, T_s)$ be as guaranteed there.
		In particular, we have $\sum_{h=1}^s |I_h| \ge (1 / p^*\alpha^*) |T^0|$.
		By Lemma~\ref{lem:intervals}, there is a subset $G \subset [s]$ such that $\sum_{h \in G}|I_h| \ge (1 / 4p^*\alpha^*) |T^0| \ge (\eps / 4p^*\alpha^*)|I| = (\eps/p\alpha)|I|$; and, for every interval $J$ in $I$ that contains an interval $I_h$ with $h \in [G]$, $\sum_{h \in [s] \colon I_h \subset J} |I_h| \ge (\eps / 4p^*\alpha^*) |J|$. Since each $I_h$ contains at least $\alpha^* |I_h|$ disjoint length-$k$ increasing subsequences, it follows that $J$ contains at least $(\eps/4p^*) |J| = (\eps/p)  |J|$ length-$k$ increasing subsequences. Taking $T$ to be the union of $T_h$ over $h \in G$, along with the pairs $(I_h, T_h)$ with $h \in G$, we obtain the required robust splittable intervals. 
	\end{proof}



\newcommand{\Event}{\mathtt{Event}}

\section{The Algorithm}


Our aim in this section is to prove the existence of a randomized algorithm, $\FindMon_k(f, \eps, \delta)$, that receives as input a function $f \colon I \to \R \cup \{*\}$ (where $I \subset \N$ is an interval), and parameters $\eps, \delta \in (0,1)$, and satisfies the following: if $f$ contains $\eps |I|$ disjoint $(12 \ldots k)$-patterns, then the algorithm outputs such a pattern with probability at least $1 - \delta$; and the running time of the algorithm is $O_{k, \eps}(\log n)$.
To this end, we describe such an algorithm in Figure~\ref{fig:find-mon} below. This algorithm uses three subroutines: $\SampleSuffix$, $\FindWithinInterval$, and $\FindGoodSplit$, the first of which is given in \cite{BCLW19}, and the latter two are described below, in Figures~\ref{fig:find-within-interval} and \ref{fig:find-good-split}. The majority of the section is devoted to the proof that $\FindMon$ indeed outputs a $(12 \ldots k)$-pattern with high probability as claimed. Specifically, we shall prove the following theorem.

\begin{theorem}\label{thm:main-alg}
	Let $k \in \N$.
	The randomized algorithm $\emph{\FindMon}_k(f, \eps, \delta)$, described in Figure~\ref{fig:find-mon}, satisfies the following.
	Given a function $f\colon I \to \R \cup \{*\}$ and parameters $\eps, \delta \in (0,1)$, if $f$ contains $\eps |I|$ disjoint $(12\dots k)$-patterns, then $\emph{\FindMon}_{k}(f, \eps, \delta)$ outputs a $(12\dots k)$-pattern of $f$ with probability at least $1-\delta$.
\end{theorem}


Our proof proceeds by induction on $k$. It relies on Lemmas~\ref{lem:sample-suffix}, \ref{lem:sample-within}, \ref{lem:good-split}, the proofs of the latter two of which assume that Theorem~\ref{thm:main-alg} holds for smaller $k$.
We first state and prove these lemmas, and then we prove Theorem~\ref{thm:main-alg}. 

To complete the picture, we need to upper-bound the query complexity and running time of $\FindMon$. We do this in the following lemma, whose proof we delay to the end of the section.

\begin{lemma} \label{lem:alg-complexity}
	Let $f \colon I \to \R \cup \{*\}$, where $I$ is an interval of length at most $n$. The query complexity and running time of $\emph{\FindMon}_k(f, \eps, \delta)$ are at most 
	\[
		\left(k^k \cdot (\log(1/\eps))^k \frac{1}{\eps} \cdot \log(1/\delta)\right)^{O(k)} \cdot \log n.
	\]
\end{lemma}


\subsection{The $\SampleSuffix$ Sub-Routine}


We re-state Lemma~3.1 from \cite{BCLW19} which gives the $\SampleSuffix_k$ subroutine, with a few adaptations to fit our needs. 

\begin{lemma}[\cite{BCLW19}]\label{lem:sample-suffix}
	Consider any fixed value of $k \in \N$, and let $C > 0$ be a large enough constant. There exists a non-adaptive and randomized algorithm, $\emph{\SampleSuffix}_k(f, \eps, \delta)$ which takes three inputs: query access to a function $f \colon I \to \R \cup \{ * \}$, where $I \subset [n]$ is an interval, a parameter $\eps \in (0, 1)$, and an error probability bound $\delta \in (0, 1)$. Suppose there exists $\alpha \in (0,1)$, and a set $H \subseteq I$ of $(\alpha, Ck \alpha)$-growing suffixes of $f$ satisfying $\alpha |H| \geq \eps |I|$. Then, $\emph{\SampleSuffix}_k(f, \eps, \delta)$ finds a length-$k$ monotone subsequence of $f$ with probability at least $1-\delta$.
	The query complexity of $\emph{\SampleSuffix}_k(f, \eps, \delta)$ is at most
	\[ 
		\frac{\log n}{\eps} \cdot \polylog(1/\eps) \cdot \log(1/\delta).
	\]
\end{lemma}

	The few adaptations that Lemma~\ref{lem:sample-suffix} has in comparison to Lemma~3.1 from \cite{BCLW19} are with respect to the error probability going from $9/10$ to $1-\delta$, and the fact that we are considering functions $f\colon I \to \R \cup \{*\}$ as opposed to $f\colon [n] \to \R$.
	
	{In order to achieve error probability $1-\delta$ in Lemma~3.1 of \cite{BCLW19}, we perform $O(\log(1/\delta))$ independent repetitions of $\SampleSuffix_k$, as described in \cite{BCLW19}. These  are reflected in the query complexity.

	The second difference is that we consider functions $f\colon I \to \R \cup \{*\}$. Inspecting the proof of Lemma~3.1 in \cite{BCLW19}, one can see that $\SampleSuffix_k$ is guaranteed to output, with high probability, $(12\dots k)$-patterns whose indices are specified in Definition~\ref{def:growing-suffixes}. Since the algorithm is non-adaptive, enforcing that indices not partaking in growing suffices not be used (by making them $*$) does not affect that analysis.

\subsection{Handling Overshooting: The $\FindWithinInterval$ Sub-Routine}


In this section, we describe the $\FindWithinInterval$ subroutine, addressing the overshooting case as explained in Section \ref{subsec:techniques}. 

\begin{figure}[ht!]
	\begin{framed}
		\begin{minipage}{.98\textwidth}
			Subroutine $\FindWithinInterval_k(f, \eps, \delta, x, y,\calJ)$. 

			\vspace{.3cm}
			{\bf Input:} Query access to a function $f \colon I \to \R \cup \{ * \}$, parameters $\eps, \delta \in (0,1)$, two inputs $x, y \in I$ where $x < y$ and $f(x) < f(y)$, and $\calJ = (J_1, \dots, J_{k-2})$ which is a collection of disjoint intervals appearing in order inside $[x, y]$. 

			\vspace{.1cm}
			{\bf Output:} a sequence $i_1 < \ldots < i_k$ with $f(i_1) < \ldots < f(i_k)$, or $\fail$.
			\begin{enumerate}
				\item 
					For every $\kappa \in [k-2]$, let $f_{\kappa}, f_\kappa' \colon J_\kappa \to \R \cup \{ * \}$ be given by:
					\begin{align}
						f_\kappa(i) &= \left\{ \begin{array}{cc} f(i) & f(i) < f(y) \\
										* 	  & \text{o.w.} \end{array} \right. \qquad\text{and}\qquad f_\kappa'(i) = \left\{ \begin{array}{cc} f(i) & f(i) \geq f(y) \\
												* &\text{o.w.}\end{array} \right. .\label{eq:capped-f}
					\end{align}

				\item
					\label{en:type-1} Call $\FindMon_{\kappa+1}(f_\kappa, \eps/2, \delta / (2k))$ for every $\kappa \in [k-2]$. 
				\item
					\label{en:type-2} Call $\FindMon_{k-\kappa}(f_\kappa' , \eps/2,  \delta / (2 k))$ for every $\kappa \in [k-2]$.
				\item 
					Consider the set of all indices that are output in Lines~\ref{en:type-1} and \ref{en:type-2}, together with $x$ and $y$. If there is a length-$k$ increasing subsequence among these indices, output it. Otherwise, output \fail.
			\end{enumerate}
		\end{minipage}
	\end{framed}\vspace{-0.2cm}
	\caption{Description of the $\FindWithinInterval$ subroutine.}
	\label{fig:find-within-interval}
\end{figure}

\begin{lemma}\label{lem:sample-within}
	Consider the randomized algorithm, $\emph{\FindWithinInterval}_{k}(f, \eps, \delta, x, y, \calJ)$, described in Figure~\ref{fig:find-within-interval}, which takes six inputs:
	\begin{itemize}
		\item 
			Query access to a function $f \colon I \to \R \cup \{ * \}$,
		\item 
			Two parameters $\eps, \delta \in (0, 1)$,
		\item 
			Two points $x, y \in I$ where $x < y$ and $f(x) < f(y)$, and
		\item 
			A collection $\calJ = (J_1, \dots, J_{k-2})$ of $k-2$ disjoint intervals which appear in order (i.e., $J_\kappa$ comes before $J_{\kappa+1}$) within the interval $[x, y]$,
	\end{itemize} 
	and outputs either a length-$k$ increasing subsequence of $f$, or \emph{\fail}.

	Suppose that for every $\kappa \in [k-2]$, the function $f|_{J_\kappa} \colon J_\kappa \to \R \cup \{*\}$, contains $\eps|J_\kappa|$ disjoint $(12\dots k)$-patterns. Then, assuming that Theorem~\ref{thm:main-alg} holds for every $k'$ with $1 \le k' < k$, $\emph{\FindWithinInterval}_k(f, \eps, \delta, x, y, \calJ)$ outputs a length-$k$ monotone subsequence of $f$ with probability at least $1-\delta$. 
\end{lemma}


\begin{proof}
	For each $\kappa \in [k-2]$, let $\calC_{\kappa}$ be a collection of at least $\eps |J_{\kappa}|$ disjoint $(12 \ldots k)$-patterns in $J_{\kappa}$. We form the following two collections, of suffixes and prefixes of $(12 \ldots k)$-patterns in $\calC_{\kappa}$.  
	\begin{align*}
		\calA_{\kappa} &= 
		\{(i_1, \ldots, i_{\kappa+1}) : (i_1, \ldots, i_{\kappa+1}) \text{ is a prefix of a $k$-tuple from $\calC_k$, and } f(i_{\kappa+1}) < f(y)\} \\
		\calB_{\kappa} &= 
		\{(i_{\kappa+1}, \ldots, i_k) : (i_{\kappa+1}, \ldots, i_k) \text{ is a suffix of a $k$-tuple from $\calC_k$, and } f(i_{\kappa+1}) \ge f(y)\}
	\end{align*}
	Note that for each $(12 \ldots k)$-pattern in $\calC_{\kappa}$, either its $(\kappa+1)$-prefix is in $\calA_{\kappa}$, or its $(k-\kappa)$-suffix is in $\calB_{\kappa}$. Thus, at least one of $\calA_{\kappa}$ and $\calB_{\kappa}$ has size at least $(\eps/2)|J_{\kappa}$. Say that $J_{\kappa}$ is of type-$1$ if $|\calA_{\kappa}| \ge (\eps/2)|J_{\kappa}|$, and otherwise say that $J_{\kappa}$ is of type-$2$ (in which case $|\calB_{\kappa}| \ge (\eps/2)|J_{\kappa}|$).

	Now, if $J_{\kappa}$ is of type-$1$, then Line~\ref{en:type-1}, called with $\kappa$, will find a $(12\dots(\kappa+1))$-pattern with probability at least $1 - \delta / (2k)$, by Theorem~\ref{thm:main-alg} for $\kappa+1 < k$ (namely, the inductive hypothesis) and the lower bound on $|\calA_{\kappa}|$.
	On the other hand, if $J_{\kappa}$ is of type-$2$, Line~\ref{en:type-2} will output a $(12\dots (k-\kappa))$-pattern with probability at least $1-\delta / (2k)$, due to the inductive hypothesis and the lower bound on $|\calB_{\kappa}|$. Thus, by a union bound, with probability at least $1 - \delta$, Line~\ref{en:type-1} outputs a pattern whenever $J_{\kappa}$ is of type-$1$, and Line~\ref{en:type-2} outputs a pattern whenever $J_{\kappa}$ is of type-$2$.


	Notice that if $J_1$ is of type-2, the $(12\dots (k-1))$-pattern returned in Line~\ref{en:type-2} can be combined with $x$ to form a $(12\dots k)$-pattern. Hence, we may assume that $J_1$ is of type-$1$. Furthermore, if $J_{k-2}$ is of type-$1$, the $(12\dots (k-1))$-pattern found in Line~\ref{en:type-1} can be combined with $y$ to form a $(12\dots k)$-pattern, and hence, we may assume that $J_{k-2}$ is of type-$2$. Thus, there exists some $\kappa \in [k-3]$ where $J_{\kappa}$ is of type-$1$ and $J_{\kappa + 1}$ is of type-$2$. Since $J_{\kappa}$ comes before $J_{\kappa+1}$, and since non-starred elements in $f_{\kappa}$ lie below the non-$*$ elements of $f_{k+1}'$, we can combine the $(12\dots (\kappa+1))$-pattern in $f_{\kappa}$ with the $(12\dots (k - \kappa-1))$-pattern in $f_{\kappa + 1}'$. 
\end{proof}

\subsection{Handling the Fitting Case: The $\FindGoodSplit$ Sub-Routine}


In this section, we describe the $\FindGoodSplit$ subroutine, which corresponds to the fitting case from Section \ref{subsec:techniques}.

\begin{figure}[ht!]
	\begin{framed}
		\begin{minipage}{.98\textwidth}
			Subroutine $\FindGoodSplit_k(f, \eps, \delta, c, \xi)$. 

			\vspace{0.3cm}
			
			{\bf Input:} Query access to a function $f \colon I \to \R \cup \{ * \}$, parameters $\eps, \delta \in (0,1)$, and $c \in [k-1]$. We let $c_1 > 1$ be a large enough (absolute) constant. 
			
			\vspace{.1cm}
			{\bf Output:} a sequence $i_1 < \ldots < i_k$ with $f(i_1) < \ldots < f(i_k)$, or $\fail$.
			\begin{enumerate}
			\item\label{en:line1} Repeat the following procedure $t = c_1 k / (\eps \xi^2) \cdot \log(1/\delta)$ times:
			\begin{enumerate}
			\item Sample $\bw, \bz \sim I$, and consider the functions $f_{\bz, \bw} \colon I \cap (-\infty, \bz)  \to \R \cup \{ * \}$ and $f_{\bz, \bw}' \colon I \cap [\bz, \infty) \to \R \cup \{ * \}$ given by 
			\begin{align}
			f_{\bz, \bw}(i) &= \left\{\begin{array}{cc} f(i) & f(i) < f(\bw) \\
											 * & \text{o.w.}\end{array} \right. \qquad\text{and}\qquad f_{\bz, \bw}' (i) = \left\{ \begin{array}{cc} f(i) & f(i) \geq f(\bw) \\
															* & \text{o.w.} \end{array} \right. .
			\end{align}
			\item\label{en:line-1b} Run $\FindMon_{c}(f_{\bz,\bw}, \eps\xi/3, \delta/3)$ and $\FindMon_{k-c}(f_{\bz,\bw}' , \eps \xi/3, \delta/3)$.
			\end{enumerate} 
			\item If we ever find a length-$k$ monotone subsequence of $f$, output it, otherwise, output $\fail$.
			\end{enumerate}
		\end{minipage}
	\end{framed}\vspace{-0.2cm}
	\caption{Description of the $\FindGoodSplit$ subroutine.}
	\label{fig:find-good-split}
\end{figure}

\begin{lemma}\label{lem:good-split}
	Consider the randomized algorithm $\emph{\FindGoodSplit}_{k}(f, \eps, \delta, c, \xi)$, described in Figure~\ref{fig:find-good-split}, which takes as input five parameters:
	\begin{itemize}
		\item 
			Query access to a function $f \colon I \to \R \cup \{ * \}$,
		\item 
			Two parameters $\eps, \delta \in (0, 1)$,
		\item 
			An integer $c \in [k-1]$, and
		\item 
		A parameter $\xi \in (0,1]$,
	\end{itemize}
	and outputs either a length-$k$ increasing subsequence or \emph{\fail}.

	Suppose that there exists an interval-tuple pair $(I', T)$ which is $(c, 1/(6k), \eps)$-splittable and $|I' | / |I| \geq \xi$. Then, the algorithms $\emph{\FindGoodSplit}_k(f, \eps, \delta, c, \xi)$ finds a $(12\dots k)$-pattern of $f$ with probability $1-\delta$. 
\end{lemma}



\begin{proof}
	Let $(I', T)$ be $(c, 1/(6k), \eps)$-splittable, and let $L, M, R$ be the contiguous intervals splitting $I'$ as in Definition~\ref{def:splittable}. Furthermore, let $T^{(L)}$ and $T^{(R)}$ be as in Definition~\ref{def:splittable}. Writing 
	\begin{align*}
		m_{1} &= \rank\left(\left\{ f(i_{c}) : (i_1, \dots, i_c) \in T^{(L)} \right\}, |T|/3\right), \\
		m_{2} &= \rank\left( \left\{ f(i_c) : (i_1,\dots, i_c) \in T^{(L)} \right\}, 2|T| /3 \right),
	\end{align*}
	%
	as the $(|T|/3)$-largest and $(2|T|/3)$-largest elements in $\left\{ f(i_c) : (i_1,\dots, i_c) \in T^{(L)} \right\}$ (taking multiplicity into account). Let $T^{(L)}_l$ be the $(12\dots c)$-patterns in $T^{(L)}$ where the $c$-th index is at most $m_{1}$, and $T^{(R)}_h$ be the $(k-c)$-patterns in $T^{(R)}$ whose $(c+1)$-th index is larger than $m_{2}$. Notice that $|T^{(L)}_l|, |T^{(R)}_h| \geq |T|/3$, and that any $(12\dots c)$-pattern from $T^{(L)}_l$ can be combined with any $(12\dots(k-c))$-pattern from $T^{(R)}_h$ to form a $(12\dots k)$-pattern. Furthermore, there exists  $|T|/3$ indices in $I'$ whose function value lies in $[m_1, m_2]$.

	Consider the event, defined over the randomness of $\bw, \bz \sim I$ that: $\bz \in M$; and $\bw$ satisfies $f(\bw) \in [m_{1}, m_{2}]$. This event occurs at some iteration of Line~\ref{en:line1}, with probability at least $1-\delta /3$; this is because there are $|M| \geq |I'| / (6k) \geq (\xi /(6k)) |I|$ valid indices for $\bz$, and there are at least $|T|/3 \geq (\eps/3) |I'| \geq (\eps \xi / 3) |I|$ indices for $\bw$, so the probability that the pair $(\bz, \bw)$ satisfies the requirements is at least $\eps \xi^2 / (18k)$. We obtain the desired bound by the setting of $t$, since $c_1$ is set to a large enough constant.

	Notice that when this event occurs, the $(12\dots c)$-patterns in $T^{(L)}_l$ all lie in $f_{\bz, \bw}$, and the $(12\dots(k-c))$-patterns in $T^{(R)}_h$ all lie in $f_{\bz, \bw}'$. Thus, $f_{\bz, \bw}$ contains at least $|T|/3 \geq (\eps/3) |I'| \geq (\eps \xi/3) |I|$ disjoint $(12\dots c)$-patterns, and $f_{\bz, \bw}'$ similarly contains at least $(\eps \xi / 3)|I|$ disjoint $(12\dots (k-c))$-patterns. Thus, by the inductive hypothesis, Line~\ref{en:line-1b} finds a $(12\dots c)$-pattern in $f_{\bz, \bw}$ and a $(12\dots (k-c))$-pattern in $f_{\bz, \bw}'$ with probability at least $1-2\delta /3$, and these can be combined to give a $(12\dots k)$-pattern of $f$. 
\end{proof}

\subsection{The Main Algorithm}

\begin{figure}[ht!]
	\begin{framed}
		\begin{minipage}{.98\textwidth}
			Subroutine $\FindMon_k(f, \eps, \delta)$. 

			\vspace{.3cm}
			{\bf Input:} Query access to a function $f \colon I \to \R \cup \{ * \}$, parameters $\eps, \delta \in (0,1)$. We let $c_1, c_2, c_3 > 0$ be large enough constants, and let $p = P(k\log(1/\eps))$, where $P \colon \R \to \R$ is a polynomial of large enough (constant) degree.

			\vspace{.1cm}
			{\bf Output:} a sequence $i_1 < \ldots < i_k$ with $f(i_1) < \ldots < f(i_k)$, or $\fail$.
			\begin{enumerate}
				\item
					\label{main-alg:line1} Run $\SampleSuffix_{k}(f, \eps / p, \delta)$.
				\item
					\label{main-alg:line2} Repeat the following for $c_1\log(1/\delta) \cdot p \cdot k^5 / \eps^2$ many iterations:
				
					\begin{enumerate}
						\item \label{main-alg:line2a} 
							Sample $\bx \sim I$ uniformly at random. If $f(\bx) = *$, proceed to the next iteration.
							Otherwise, if $k = 1$ output $\bx$ and proceed to Step~\ref{main-alg:line3}, and if $k \ge 2$ proceed to the next step.
						\item \label{main-alg:line2b} 
							For each $t \in [\log n]$, sample $\by_t \sim [\bx + 2^{t}/(12k),  \bx + 2^{t}]$ uniformly at random. If there exists at least one $t$ where $f(\by_t) > f(\bx)$, set
							\begin{align}
								\by &= \max \left\{ \by_t : t \in [\log n] \text{ and }f(\by_t) > f(\bx) \right\}, \label{eq:y-def}
							\end{align}
							let $t^* \in [\log n]$ be the index for which $\by_{t^*} = \by$, and continue to the next line.
							Otherwise, i.e.\ if $f(\by_t) \not> f(\bx)$ for every $t$, continue to the next iteration.
						\item \label{main-alg:line2c}
								If $k = 2$, output $(\bx, \by)$ and proceed to Step~\ref{main-alg:line3}. If $k > 2$, continue to the next line.

							\item 
								 Here $k \geq 3$. Set $\ell = 4p / \eps$ and perform the following.
				
								\begin{enumerate} 
									\item\label{main-alg:line2dii} 
										Consider the collection $\calJ$ of $k-2$ intervals $J_1, \dots, J_{k-2}$ appearing in order within $[\bx, \by]$, given by setting, for every $i \in [k-2]$,
										\begin{align} 
											J_i &= \left[\bx + \frac{2^{t^*}}{12k} \cdot \ell^{-(k-1-i)}, \bx + \frac{2^{t^*}}{12k} \cdot \ell^{-(k-2-i)}\right), \label{eq:intervals}
										\end{align}
						
										and run $\FindWithinInterval_k(f, \eps/2p, \delta/2,\bx, \by, \calJ)$. 
									\item \label{main-alg:line2di} 
										For each $t' \in [t^* - 3k \log \ell, t^*]$ do the following.
										\vspace{.15cm}

										Consider the interval $J_{t'} = [\bx - 2^{t'}, \bx + 2^{t'}]$, and the restricted function $g_{t'} \colon J_{t'} \to \R \cup \{*\}$ given by $g_{t'} = f|_{J_{t'}}$. For every $c_0 \in [k-1]$, run $\FindGoodSplit_k(g_{t'}, \eps / c_2, \delta/ 2, c_0, 1/4)$.
								\end{enumerate}
					\end{enumerate}
				\item \label{main-alg:line3}
					If a length-$k$ monotone subsequence of $f$ is found, output it. Otherwise, output \fail.
			\end{enumerate}
		\end{minipage}
	\end{framed}\vspace{-0.2cm}
	\caption{Description of the $\FindMon_k$ subroutine.}
	\label{fig:find-mon}
\end{figure}

%
%
%
%

Consider the description of the main algorithm in Figure~\ref{fig:find-mon}.
We prove Theorem~\ref{thm:main-alg} by induction on $k$. The proof uses Lemma~\ref{lem:sample-suffix}, Lemma~\ref{lem:sample-within}, and Lemma~\ref{lem:good-split}. 

\begin{proof}[Proof of Theorem~\ref{thm:main-alg}]
	\hfill
	
	\paragraph{Base Case: \normalfont{$k = 1$}.} 
	
		\hfill

		Recall that $f$ has at least $\eps |I|$ non-$*$ values. Thus, with probability at least $1 - \delta$, a non-$*$ value is observed after sampling $\bx \sim I$ at least $(1/\eps) \cdot \log(1/\delta)$ times. It follows that with probability at least $1 - \delta$, Line~\ref{main-alg:line2a} of our main algorithm, given in Figure~\ref{fig:find-mon}, samples $\bx \neq *$ in one of its iterations.

		\paragraph{Inductive Step: \normalfont{proof of Theorem~\ref{thm:main-alg} for $k \ge 2$, under the assumption that it holds for every $k'$ with $1 \le k' < k$}.}

		\hfill

		
		Let $p = P(k \log(1/\eps))$ (recall that $P(\cdot)$ is a polynomial of sufficiently large (constant) degree). Apply Theorem~\ref{thm:main-structure-2} to $f$. 
	
		Suppose, first, that (\ref{en:suffix}) of Theorem~\ref{thm:main-structure-2} holds. So, there exists a set $H \subset [n]$ of indices that start an $(\alpha, Ck \alpha)$-growing suffix, with $\alpha |H| \geq (\eps / p) n$, for some $\alpha \in (0,1)$. By Lemma~\ref{lem:sample-suffix}, the call for $\SampleSuffix_k(f, \eps / p, \delta)$ in Line~\ref{main-alg:line1} outputs a length-$k$ monotone subsequence of $f$ with probability at least $1-\delta$. 	
		
		Now suppose that (\ref{en:split}) of Theorem~\ref{thm:main-structure-2} holds, and let $(I_1,T_1), \dots, (I_s, T_s)$ be a $(c, 1/(6k), \alpha)$-splittable collection for some $\alpha \ge \Omega( \eps / k^5)$  and $c \in [k-1]$, satisfying (\ref{eqn:splittable-intervals-robust}) and, moreover, that any $J \subset I$ with $J \supset I_h$ for some $h \in [s]$ contains $(\eps/p) |J|$ disjoint $(12\dots k)$-patterns.
		Let $\Event$ be the event that, for a particular iteration of Lines~\ref{main-alg:line2a} and \ref{main-alg:line2b}, $\bx$ is the $1$-entry of some $k$-tuple from $T_h$, for some $h \in [s]$, and $\by_t$ is the $(c+1)$-entry of some (possibly other) $k$-tuple in $T_h$, where $t$ is such that $|I_h| \le 2^t < 2|I_h|$.

%

		\begin{claim}
			$\Prx[\Event] \ge \eps \alpha / (2p)$.
		\end{claim}

		\begin{proof}
			For each $h \in [s]$, let $A_h$ and $B_h$ be the collections of $1$- and $(c+1)$-entries of patterns in $T_h$. Then
			\[ 
				\sum_{h=1}^s |A_h| 
				= \sum_{h=1}^s |T_h| 
				\geq \alpha \sum_{h=1}^s |I_h| 
				\geq \frac{\eps}{p} \cdot |I|.
			\]
			The first inequality follows from the assumption that $(I_h, T_h)$ is $(c, 1/(6k), \alpha)$-splittable, and the second inequality follows from the assumption that (\ref{eqn:splittable-intervals-robust}) holds. 

			As a result, the probability over the draw of $\bx \sim I$ in Line~\ref{main-alg:line2a} that $\bx \in A_h$ is at least $\eps / p$. Fix such an $\bx$, and consider $t \in [\log n]$ for which $|I_h| \leq 2^t < 2|I_h|$. Notice that $B_h \subset [\bx + 2^{t} / (12k), \bx + 2^t]$ since $2^{t-1} \leq |I_h| < 2^t$, and that the distance between any index of $A_h$ and $B_h$ is at least $|I_h| / (6k) \ge 2^t/(12k)$ since $(I_h, T_h)$ is $(c, 1/(6k), \alpha)$-splittable. Therefore, the probability over the draw of $\by_t \sim [\bx + 2^{t}/(12k), \bx + 2^t]$ that $\by_t \in B_h$ is at least $|B_h| / 2^t \ge |T_h| / (2|I_h|) \geq \alpha / 2$. 
		\end{proof}

		By the previous claim, since we have $c_1 \cdot \log(1/\delta) \cdot p \cdot k^5 / \eps^2$ iterations of Lines~\ref{main-alg:line2a} and \ref{main-alg:line2b}, with probability at least $1 - \delta/2$, $\Event$ holds in some iteration (using the lower bound $\alpha \ge \Omega(\eps / k^5)$ and the choice of $c_1$ as a large constant). 

		Consider the first execution of Line~\ref{main-alg:line2a} and Line~\ref{main-alg:line2b} where $\Event$ holds (assuming such an execution exists). Let $h \in [s]$ and $t \in [\log n]$ be the corresponding parameters, i.e., $h$ and $t$ are set so $\bx$ is the first index of a $k$-tuple in $T_{h}$, $\by_t$ is the $(c+1)$-th index in another $k$-tuple in $T_h$, and $|I_h| \leq 2^t < 2|I_h|$. We consider this iteration of Line~\ref{main-alg:line2}, and assume that $\Event$ holds with these parameters for the rest of the proof. Notice that $\by$, as defined in (\ref{eq:y-def}), satisfies $\by \geq \by_t$ (as $f(\by) > f(\bx)$) and hence $t^* \geq t$. 
		
		Note that if $k = 2$, the pair $(\bx, \by)$, which is a $(12)$-pattern in $f$, is output in Line~\ref{main-alg:line2c}, so the proof is complete in this case. From now on, we assume that $k \ge 3$.
		We break up the analysis into two cases: $t^* \ge t + 3k\log \ell$ and $t^* < t + 3k \log \ell$.



		Suppose $t^* \ge t + 3k \log \ell$. We now observe a few facts about the collection $\calJ$ specified in (\ref{eq:intervals}). First, notice that $J_1, \dots, J_{k-2}$ appear in order from left-to-right, and they lie in $[\bx, \by]$ (as $\by = \by_{t^*} \in [\bx + 2^{t^*}/(12k), 2^{t^*}]$). Second, in the next claim we show that for every $i \in [k-2]$, the interval $J_i$ contains $(\eps/2p) |J_i|$ disjoint $(12\dots k)$-patterns. 
		\begin{claim}
			$J_i$ contains $(\eps/2p) |J_i|$ disjoint $(12\dots k)$-patterns.
		\end{claim}

		\begin{proof}
			Let $J_i'$ be the interval given by
			\[ 
				J_i' = I_h \cup \left[ \bx, \bx + \frac{2^{t^*}}{12k} \cdot \ell^{-(k-2-i) } \right].  
			\]
			Observe that 
			\[
				|J_i' \setminus J_i| \le 
				2^{t} + \frac{2^{t^*}}{12k} \cdot \ell^{-(k-1-i)} \le 
				\frac{2^{t^*}}{6k} \cdot \ell^{-(k-1-i)} = 
				\frac{2}{\ell} \cdot \frac{2^{t^*}}{12k} \cdot \ell^{-(k-2-i)} \ge 
				\frac{2}{\ell} \cdot |J_i'| =
				\frac{\eps}{2p} \cdot |J_i'|,
			\]
			where for the second inequality we used the bound $t^* - t \ge 3 k \log \ell \ge  \log(12) + \log k + (k-2)\log \ell$, and that $\ell = 4p / \eps$.
			
			We have by Theorem~\ref{thm:main-structure-2}, that $J_i'$ contains at least $(\eps/p) |J_i'|$ disjoint $(12\dots k)$-patterns in $f$. Hence, the number of disjoint $(12\dots k)$-patterns in $J_i$ is at least:
			\begin{align*}
				\frac{\eps}{p} \cdot |J_i' | - |J_i' \setminus J_i| 
				\geq \frac{\eps}{2p} \cdot |J_i'| 
				\geq \frac{\eps}{2p} \cdot |J_i|,
			\end{align*}
			as required.
		\end{proof}
		
		By Lemma~\ref{lem:sample-within}, Line~\ref{main-alg:line2dii} outputs a $(12\dots k)$-pattern in $f$ with probability at least $1-\delta /2$. By a union bound, we obtain the desired result.

		Suppose, on the other hand, that $t^* \le t + 3k\log \ell$. In this case, as $2^{t-1} \le |I_h| \le 2^{t^*}$ (by choice of $t$), for one of the values of $t'$ considered in Line~\ref{main-alg:line2di} we have $2^{t'-1} \le |I_h| < 2^{t'}$; fix this $t'$. The interval $J_{t'}$, defined in Line~\ref{main-alg:line2di}, hence satisfies $|I_h| / |J_{t'}| \geq 1/4$. As a result, and since $I_h \subset J_{t'}$ (because $t \le t^*$), the function $g \colon J \to \R \cup \{*\}$ contains an interval-tuple pair $(I_h, T_h)$ which is $(c, 1/(6k), \alpha)$-splittable. By Lemma~\ref{lem:good-split}, once Line~\ref{main-alg:line2di} considers $c_0 = c$, the sub-routine $\FindGoodSplit_k(g, \eps / (c_2 k^5), \delta/2, c, 1/4)$ will output a $(12\dots k)$-pattern of $g_{t'}$ (which is also a $(12\dots k)$-pattern of $f$) with probability at least $1 - \delta / 2$. Hence, we obtain the result by a union bound.
\end{proof}

\subsection{Query Complexity and Running Time}

	It remains to prove Lemma~\ref{lem:alg-complexity}, estimating the number of queries made by \FindMon, as well as its total running time.

	\begin{proof}[Proof of Lemma~\ref{lem:alg-complexity}]
		We first claim that the running time is bounded by an expression of the form $\poly(k)$ times the query complexity of \FindMon, where the $\poly(\cdot)$ term is of constant degree. Indeed, the only costly operations (in terms of running time) other than querying that our algorithm conducts involve:
		\begin{itemize}
			\item Determining whether the value of $f$ at a certain point is $\ast$ or not; to this end, note that for any $f$ we need to evaluate along the way, $f(x)$ is marked by $\ast$ if and only if it does not belong to some interval in $\R$, whose endpoints are determined by the recursive calls that led to it. Since the recursive depth is at most $k$, this means that the complexity of the above operation is $O(k)$.\footnote{In fact, this complexity can be improved to $O(1)$ if, instead of working with functions of the form $f \colon I \to \R \cup \{\ast\}$, we would have worked with function $f \colon I \to \R$ and received the interval of ``non-$*$ values'' as an input to the recursive call.}
			\item Given an ordered set of queried elements $Q$ at some point along the algorithm, determining whether these elements contain a $c$-increasing subsequence for $c \leq k$ (this action is taken, e.g., in the last part of $\FindMon$). This operation can be implemented in time $O(c|Q|)$. Now, the number of such operations that each queried element participates in is at most $k$,\footnote{More precisely, for the purpose of this section, if an element is queried $t > 1$ times by our algorithm then we think of it as contributing $t$ to the total query complexity (since our goal is to prove upper bounds here -- not lower bounds -- this perspective is clearly valid); and in this case, the number of operations as above in which it participates is at most $k \cdot t$.} and a simple double counting argument implies that the running time of these operations altogether is at most $O(k^2)$ times the total query complexity.
		\end{itemize}

		It remains now to prove the bound on the query complexity. 
		Recall that $P \colon \R \to \R$ is a fixed polynomial; write $p_{k,\eps} = P(k\log(1/\eps))$. We fix $n$, which upper bounds the length of all intervals defining input functions. Let $\Phi(k, \eps, \delta)$ be the maximum number of queries made by $\FindMon_k(f, \eps, \delta)$. Let 
		\begin{align*}
			\Phi^{(1)}(k, \eps, \delta) &= 
				\begin{array}{l} 
					\text{query complexity of $\SampleSuffix_{k}(f, \eps, \delta)$.}
				\end{array}  \\
			\Phi^{(2)}(k, \eps, \delta) &= 
				\begin{array}{l} 
					\text{query complexity of $\FindWithinInterval_{k}(f, \eps, \delta, x, y, \calJ)$,}\\
					\text{where $|\calJ| = k-2$.} 
				\end{array} \\
			\Phi^{(3)}(k, \eps, \delta, \xi) &= 
				\begin{array}{l} 
					\text{query complexity of $\FindGoodSplit_k(f, \eps, \delta, c, \xi)$,}\\
					\text{where $c \in [k-1]$.} 
				\end{array}
		\end{align*}
		By Lemma~\ref{lem:sample-suffix}, as well as an inspection of Figure~\ref{fig:find-within-interval} and Figure~\ref{fig:find-good-split}, we have:
		\begin{align*}
			\Phi^{(1)}(k, \eps, \delta) &\leq p_{k,\eps} \cdot \frac{1}{\eps} \cdot \log(1/\delta) \cdot \log n \\
			\Phi^{(2)}(k, \eps, \delta) &\leq 2k \cdot \Phi(k-1, \eps/2, \delta/(2k)) \\
			\Phi^{(3)}(k, \eps, \delta, \xi) &\leq \frac{c_1 k \log(1/\delta)}{\eps \xi^2} \cdot \Phi(k-1, \eps \xi/3, \delta/3).
		\end{align*}
		Lastly, inspecting Figure~\ref{fig:find-mon}, we have
		\begin{align*}
			\Phi(k, \eps, \delta) &\leq \Phi^{(1)}(k, \eps/p_{k,\eps}, \delta) + \\
			&\qquad\qquad c_1 \cdot p_{k,\eps} \cdot \frac{k^5}{\eps^2} \cdot \log(1/\delta) \cdot
			\left(1 + \log n + \Phi^{(2)}\left(k, \eps / (2p_{k,\eps}), \delta/2\right) + 
			\Phi^{(3)}\left(k, \eps/(c_2 k^5), \delta/2, 1/4\right)\right) \\
			&\leq q_{k,\eps} \cdot \frac{1}{\eps^2} \cdot \log(1/\delta) \cdot \log n \, + q_{k,\eps} \cdot \frac{1}{\eps^3} \cdot (\log(1/\delta))^2 \cdot  \Phi(k-1, \eps/q_{k,\eps}, \delta/(3k)) \\
			&\leq \left(k^k \cdot (\log(1/\eps))^k \cdot \frac{1}{\eps} \cdot \log (1/\delta) \right)^{O(k)} \cdot \log n,
		\end{align*}
		where $Q \colon \R \to \R$ is a fixed polynomial of large enough (constant) degree and $q_{k,\eps} = Q(k \log(1/\eps))$. For the last line we use that $\Phi^{(2)}(1, \cdot, \cdot) = \Phi^{(2)}(2, \cdot, \cdot) = \Phi^{(3)}(1,\cdot,\cdot ,\cdot) = \Phi^{(3)}(2, \cdot,\cdot,\cdot)=0$, and we note that the parameter replacing $\eps$ never falls below $\eps / (k\log(1/\eps))^{O(k)}$, so the factor of $\log n$ at each iteration is at most $\left(k^k (\log(1/\eps))^k (1/\eps) \log(1/\delta)\right)^{O(k)}$.
	\end{proof}

\ignore{
\section{Old manuscript}
	\begin{theorem} \label{thm:new-structure}
		Let $k \in \N$, $\eps > 0$, let $f : [n] \to \R$ be a function, and let $T_0$ be a set of at least $\eps n$ disjoint $(12 \ldots k)$-patterns. Then one of the following holds, where $C, c > 0$ are constants. Then there exists $\alpha \ge \eps / \poly(k)$ such that one of the following conditions holds.
		\begin{itemize}
			\item
				Either there exists a set $H \subseteq [n]$ of indices that start $(\alpha, C k \alpha)$-growing suffixes, satisfying 
				\[
					\alpha |H| \ge \frac{\eps}{\poly(k, \log(1/\eps))}.
				\]
			\item
				Or, there exists $l \in [k-1]$, a set $T$, with $E(T) \subseteq E(T_0)$, of disjoint $(12 \ldots k)$-patterns, and a $(l, 1/(6k), \alpha)$-splittable collection of $T$, of disjoint interval-tuple pairs $(I_1, T_1), \ldots, (I_s, T_s)$, such that
				\[
					\alpha \sum_{h=1}^s |I_h| \ge \delta n,
				\]
				where $\delta \ge \eps / \poly(k, \log(1/\eps))$.
				Moreover, for every interval $J$ that contains an interval $I_h$, with $h \in [s]$, $J$ contains at least $c \cdot \delta |J|$ disjoint $(12 \ldots k)$-patterns whose elements are in $E(T_0)$.
		\end{itemize}
	\end{theorem}

	\begin{proof}
		We use Theorem 2.2 from the previous paper. If the first item there holds, we are done. So suppose that the second item holds. This means that there exist an integer $l \in [k-1]$, a set $T$, with $E(T') \subseteq E(T_0)$, of disjoint length-$k$ increasing subsequences, and a $(l, 1/(6k), \alpha)$-splittable collection $(I_1, T_1), \ldots, (I_s, T_s)$ of $T'$ satisfying
		\[
			\alpha \sum_{h=1}^s |I_h| \ge \delta n.	
		\]
		Let $\calI$ be the set of intervals $I \in \{I_1, \ldots, I_s\}$ for which there is an interval $J \supseteq I$ such that $\alpha \sum_{h \in [s]:\, I_h \subseteq J} |I_h| \le c \cdot \delta |J|$; for each such $I$ fix such an interval $J(I)$. 
		Let $\calJ$ be a minimal collection of intervals with the following two properties.
		\begin{itemize}
			\item
				for every $J \in \calJ$
				\[
					\alpha \sum_{h \in [s]: \, I_h \subseteq J} |I_h| \le c \cdot \delta |J|.
				\]
			\item
				every $I \in \calI$ is contained in some interval in $\calJ$.
		\end{itemize}
		Note that such a collection exists, because $\{J(I) : I \in \calI\}$ satisfies the above two properties (but need not be minimal). 
		\begin{claim}
			Every element $x \in [n]$ is contained in at most three intervals from $\calJ$.
		\end{claim}
		\begin{proof}
	
			Consider first an element $x$ which is covered by some interval $I \in \calI$. Let $J_{\ell}$ be an interval from $\calJ$ that contains $x$, and whose left-most element is furthest to the left among all intervals from $\calJ$ that contains $x$; pick $J_r$ similarly, for the right instead of the left; and let $J_m$ be an interval from $\calJ$ that contains $I$. We claim that $\calJ$ does not have any other intervals that contain $x$. Suppose, to the contrary, that $x \in J \in \calJ$ and $J \neq J_{\ell}, J_r, J_m$. Let $I' \in \calI$ and suppose that $I' \subseteq J$. If $I' = I$ then $I' \subseteq J_m$. If $I'$ lies to the left of $I$, then $I' \subseteq J_{\ell}$ (by the choice of $J_{\ell}$), and if it lies to the right of $I$ then $I' \subseteq J_r$. It follows that $\calJ \setminus \{J\}$ satisfies the above two items, contradicting the minimality of $\calJ$.

			Now, if $x$ is not contained in any interval of $\calI$, then we can show similarly that there are at most two intervals from $\calJ$ that contain $x$, by definiting $J_{\ell}$ and $J_r$ as above. 
		\end{proof}
		Let $U = \cup_{J \in \calJ} J$. Then, using the above claim,
		\[
			\alpha \sum_{I \in \calI} |I| 
			\le \alpha \sum_{J \in \calJ} \left(\sum_{h \in [s]: \, I_h \subseteq J} |I_h|\right)
			\le c \cdot \delta \sum_{J \in \calJ} |J|
			\le 3c \cdot \delta |U|
			\le 3c \cdot \delta n
			\le \frac{\delta n}{2}.
		\]
		Let $\calI' = \{I_h : h \in [s]\} \setminus \calI$. Then, because $\alpha \sum_{h=1}^s |I| \le \delta n / 2$, 
		\[
			\alpha \sum_{I \in \calI'} \ge \frac{\delta n}{2},
		\]
		and every interval $J$ contains at least $\alpha \sum_{h \in [s]: \, I_h \subseteq J}|I_h| \ge c \cdot \delta |J|$ disjoint $(12 \ldots k)$-patterns whose elements are in $E(T_0)$. It follows that the second item holds (with $\delta / 2$ in place of $\delta$).
	\end{proof}

	\begin{figure}[ht!]
		\begin{framed}
			\begin{minipage}{.98\textwidth}
				Subroutine $\FindMon(f, \eps, k, I, R, \mu)$.
				\vspace{.3cm}

				{\bf Input:} Query access to a function $f \colon [n] \to \R$,  parameters $\eps, \mu \in (0,1)$, an integer $k \in \N$, and intervals $R, I \subseteq \N$.
				\smallskip

				{\bf Output:} a sequence $i_1 < \ldots i_k$ with $f(i_1) < \ldots < f(i_k)$, or $\fail$.
				\smallskip
				\smallskip

				Let $\delta = \eps / \poly(k, \log(1/\eps))$, and let $c > 0$ be a small constant. Whenever a query for $f(x)$ returns a value which is not in $R$, we think of the value as $\ast$.

				\begin{enumerate}
					\item \label{round:one}
						Run $\SampleSuffix_k(f, \delta, \mu)$.
					\item \label{round:two}
						Sample $\Theta_{\mu}(1/\delta)$ elements $\bx \sim I$ uniformly at random. 

						If $k = 1$, output any element $x$ sampled here, if there exists such an element with $f(x) \neq \ast$; otherwise, output $\fail$.
					\item \label{round:three}
						For every element $x$ sampled in Round \ref{round:two}, if $f(x) \neq \ast$, and every $t \in [\log n]$, sample $\Theta_\mu(1/\eps)$ elements $\by \sim (x + 2^{t-1}, x + 2^t)$ uniformly at random.

						Denote by $S(x)$ the set of elements $y$ sampled in this round, for which $f(y) \neq \ast$ and $f(y) > f(a)$.

						If $k = 2$, and $S(x)$ is non-empty, output $(x, y)$ for any $y \in S(x)$. If $S(x)$ is empty for every $x$ sampled previously, return $\fail$.
					\item \label{round:four}
						For every $x$ as above, let $s(x) = \max S(x)$, let $\ell(x) = s(x) - x$, and let $t(x) = \floor{ \log (\ell(x))}$.
						Let $J_i(x) = (x + (c \cdot \delta / 2)^{k-1-i} \cdot \ell(x), x + (c \cdot \delta / 2)^{k-2-i} \cdot \ell(x)]$, for $i = 1, \ldots, k-2$.
						For $i = 1, \ldots, k-2$:
						\begin{itemize}
							\item
								Call $\FindMon(f,\, c \cdot \delta / 4,\, k-i,\, J_i(x),\, R \cap [f(y), \infty),\, \mu / (2 k))$. 
							\item
								Call $\FindMon(f,\, c \cdot \delta / 4,\, i+1,\, J_i(x),\, R \cap (-\infty, f(y)),\, \mu / (2 k))$.
						\end{itemize}
						Denote by $A_i(x)$ and $B_i(x)$ the events of success in the first and second calls above, respectively.

						\begin{itemize}
							\item
								If $A_1(x)$ holds, output $x$ followed by the sequence corresponding to the event $A_1(x)$.
							\item
								If $B_i(x)$ and $A_{i+1}(x)$ holds, for some $i \in [k-3]$,  output the sequence corresponding to $B_i(x)$ followed by the sequence corresponding to $A_{i+1}(x)$ to its right. 
							\item
								If $B_{k-2}(x)$ holds, output the sequence corresponding to $B_{k-2}(x)$ followed by $s(x)$.
						\end{itemize}
					\item \label{round:five}
						For every $x$ as before, and every $t \in [t(x) - \Theta(k \log(1/\delta)), t(x) + \Theta(k \log(1/\delta))]$
						\begin{itemize}
							\item
								Sample $\Theta_\mu(1)$ elements $\bz \sim (x, x + 2^t)$ uniformly at random.
							\item
								For every $z$ sampled in the previous step, sample $\Theta_\mu(1/\delta)$ elements $\bw \sim [x, z]$ uniformly at random.
							\item
								For every $z$ from two steps ago, and every $w$ sampled in the previous step, and every $l \in [k-1]$
								\begin{itemize}
									\item
										Call $\FindMon(f,\, \delta/4,\, l,\, (x, z),\, R \cap (-\infty, f(w)],\, \mu / (16))$.
									\item
										Call $\FindMon(f,\, \delta/4,\, k-l,\, (z, x + 2^t),\, R \cap (-\infty, f(w)],\, \mu / (16))$.
									\item
										If both calls are successful, return the output from the first call followed by the output from the second call.
								\end{itemize}
						\end{itemize}
					\item \label{round:six}
						If nothing was output in previous steps, output $\fail$.
				\end{enumerate}
			\end{minipage}
		\end{framed}\vspace{-0.2cm}
		\caption{Description of the $\FindMon$ subroutine.}
		\label{fig:find-monotone}
	\end{figure}

	\begin{theorem} \label{thm:algorithm-analysis}
		Let $f : [n] \to \R$ be a function, let $\eps, \mu \in (0,1)$, let $k \in \N$, and let $I \subseteq [n]$ and $R \subseteq \N$ be intervals. 
		Suppose that there is a set $T_0 \subseteq I^k$ of disjoint increasing subsequences, such that $f(x) \in R$ for every $x \in E(T_0)$ and $|T_0| \ge \eps |I|$. Then, with probability at least $1 - \mu$, \emph{$\FindMon(f, \eps, k, I, R, \mu)$} outputs a length-$k$ increasing subsequence $(i_1, \ldots, i_k)$ such that $f(i_j) \in R$ for every $j \in [k]$.
	\end{theorem}

	\begin{corollary}
		Let $f : [n] \to \R$ be a function which is $\eps$-far from $(12 \ldots k)$-free. Then the subroutine \emph{$\FindMon(f, \eps/k, k, [n], \R, 0.1)$} is an adaptive algorithm that finds a $(12 \ldots k)$-pattern, with probability at least $0.99$, and that makes at most $O( (k / \eps)^{O(k)} \cdot \log n)$ queries.
	\end{corollary}

	\begin{proof}
		The corollary follows from \Cref{thm:algorithm-analysis}, and the query complexity follows from analyzing the algorithm. 
	\end{proof}

	\begin{proof}[Proof of \Cref{thm:algorithm-analysis}]
		The proof proceeds by induction on $k$. 

		For $k = 1$, by the assumptions there are at least $\eps |I|$ elements $x \in I$ for which $f(x) \in R$. Thus, Round \ref{round:two} will be successful with probability at least $1 - \mu$.

		For $k = 2$, if the first item in \Cref{thm:new-structure} holds, then Round \ref{round:one} will be successful with probability at least $1 - \mu$.
		\comment{I added a parameter $\mu$ to $\SampleSuffix$, so as to control the success probability of the subroutine (previously we just aimed for $.99$ success probability).}
		We may thus assume that the second item holds. Let $T$, $\alpha$, $(I_1, T_1), \ldots, (I_h, T_h)$ be as in the statement of \Cref{thm:new-structure}.
		For each $h \in [s]$, let $(L_h, M_h, R_h)$ be the partition of $I_h$ corresponding to the splittability definition. Let $\calL_h$ and $\calR_h$ be the set of first and second elements, respectively, of $(12)$-patterns in $T_h$. So $\calL_h \subseteq L_h$ and $\calR_h \subseteq R_h$, $|\calR_h| = |\calR_h| \ge \alpha |I_h|$, and $\sum_{h = 1}^s |\calL_h| \ge \delta n$. It follows that, with probability at least $1 - \mu / 2$, an element of $\cup_{h=1}^s \calL_h$ is sampled in Round \ref{round:two}. Suppose that, indeed, such an element was sampled; denote it by $x$, let $h \in [s]$ be such that $x \in \calR_h$, and let $t = \floor{\log |I_h|}$. Then, with probability at least $1 - \mu / 2$, an element of $\calR_h$ is sampled in Round \ref{round:three} (with $x$ and $t$). This means that, with probability at least $1 - \mu$, the algorithm outputs a $(12)$-pattern whose elements are in $f^{-1}(R)$.

		Now suppose that $k \ge 3$. Again, if the first item in \Cref{thm:new-structure} holds then Round \ref{round:one} will be successful with probability at least $1 - \mu$. Thus, we assume that the second item in \Cref{thm:new-structure} holds. Let $l$, $T$, $\alpha$, $(I_1, T_1), \ldots, (I_h, T_h)$ be as in the statement of the theorem; denote the corresponding partition of $I_h$ by $(L_h, M_h, R_h)$. Let $\calL_h$ be the set of first elements of subsequences in $T_h$, and let $\calR_h$ be the set of last elements of subsequences in $T_h$. Then $\calL_h \subseteq L_h$, $\calR_h \subseteq R_h$, and $\sum_{h = 1}^s |\calL_h| \ge \delta n$; let $\calL_h^-$ be the set of $|\calL_h|/2$ left-most elements of $\calL_h$, and define $\calL_h^+$, $\calR_h^-$ and $\calR_h^+$ similarly.

		With probability at least $1 - \mu / 2$, an element of $\cup_{h = 1}^s \calL_h^-$ is sampled at Round \ref{round:two}; suppose that this is the case and denote such an element by $x$, and let $h$ be such that $x \in \calL_h^-$. Then, an element $y$ of $\calR_h^+$ is sampled, with probability at least $1 - \mu/2$, during Round \ref{round:three}, when run with $x$ and $t = \floor{\log(y - x)}$. Let $S(x), s(x), t(x), \ell(x)$ be as in the algorithm. Clearly, $t(x) \ge t$; we consider two cases: $t(x) \ge  t + \Omega(k \log(1/\delta))$, or $t(x) \le t + O(k \log(1 / \delta))$.

		First, suppose that $t(x) \ge t + k \log(1/\delta)$.
		\begin{claim}
			$J_i(x)$ contains a set of at least $c \cdot \delta |J_i(x)| / 2$ disjoint $(12 \ldots k)$-patterns whose elements are in $f^{-1}(R)$.
		\end{claim}

		\begin{proof}
			Let $J_i'(x) = I_h \cup [x, x + (c \cdot \delta / 4)^{k - 2 - i} \cdot \ell(x)]$ for $i \in [s]$, and let $J_0'(x) = I_h$.
			\begin{align*}
				|J_1'(x)| = (c \cdot \delta / 2)^{k-3} \cdot \ell(x) 
				&\ge (c \cdot \delta / 2)^k \cdot 2^{\floor{\log (\ell(x))}} \\
				&= (c \cdot \delta / 2)^k \cdot 2^{t(x)} \\
				&\ge (c \cdot \delta / 2)^k  \cdot 2^t \cdot 2^{\Omega(k \log(1/\delta))} \\
				&\ge (c \cdot \delta / 2)^k  \cdot 2^{\Omega(k \log(1/\delta))} \cdot \frac{|I_h|}{6k} 
				\ge \frac{c \cdot \delta}{2} \cdot |I_h| 
				= \frac{c \cdot \delta}{2} \cdot |J_0'(x)|. 
			\end{align*}
			By assumption, $J_i'(x)$ contains $c \cdot \delta |J_i'(x)|$ disjoint $(12 \ldots k)$-patterns whose elements are in $E(T_0) \subseteq f^{-1}(R)$, at most $|J_{i-1}'(x)|$ of them intersect $J_{i-1}'(x)$. It follows that $J_i(x) = J_i'(x) \setminus J_{i-1}'(x)$ contains at least $c \cdot \delta |J_i'(x)| - |J_{i-1}'(x)| \ge c \cdot \delta |J_i(x)| / 2$ disjoint $(12 \ldots k)$-patterns whose elements are in $f^{-1}(R)$.
		\end{proof}

		Let $\calT_i$ be a set of disjoint $(12 \ldots k)$-patterns whose elements are in $J_i(x) \cap f^{-1}(R)$. We form sets $\calA_i(x)$ and $\calB_i(x)$ as follows. For each $(12 \ldots k)$-pattern $(a_1, \ldots, a_k)$, if $a_{i+1} \ge f(s(x))$, put $(a_{i+1}, \ldots, a_k)$ in $\calA_i(x)$; otherwise, put $(a_1, \ldots, a_{i+1})$ in $\calB_i(x)$. 
		Since $|\calA_i(x)| + |\calB_i(x)| = |\calT_i(x)|$, either $|\calA_i(x)| \ge c \cdot \delta |J_i(x)| / 4$ or $|\calB_i(x)| \ge c \cdot \delta |J_i(x)| / 4$. By induction, if the former holds then, then with probability at least $1 - \mu / (2k)$, the subroutine $\FindMon(f,\, c \cdot \delta / 4,\, k-i,\, J_i(x),\, R \cap [f(y), \infty),\, \mu / (2 k))$ outputs a length-$(k-i)$ increasing subsequence whose elements are in $J_i(x) \cap f^{-1}(R)$. Again, by induction, with probability at least $1 - \mu / (2k)$, the subroutine $\FindMon(f,\, c \cdot \delta / 4,\, i+1,\, J_i(x),\, R \cap (-\infty, f(y)),\, \mu / (2 k))$ outputs a length-$(i+1)$ increasing subsequence with elements in $J_i(x) \cap f^{-1}(R)$. Thus, regardless, either $A_i(x)$ or $B_i(x)$ holds, with probability $1 - \mu / (2k)$. Thus, with probability at least $1 - \mu / 2$, either $A_i(x)$ or $B_i(x)$ holds for each $i \in [k-2]$. It follows that, with probability at least $1 - \mu$, Round \ref{round:four} outputs a length-$k$ subsequence, which can easily be seen to be an increasing subsequence, whose elements are in $f^{-1}(R) \cap I$.

		Next, we consider the case where $t(x) \le t + O(k \log(1/\delta))$.
		Since $|M_h| \ge |I_h| / (6k) \ge (y-x) / (6k)$, with probability at least $1 - \mu / 4$, an element $z \in M_h$ is sampled in the first part of Round \ref{round:five}.
		Let $T_h^{(L)}$ be the set of prefixes of length $l$ of subsequences in $T_h$, and let $T_h^{(R)}$ be the set of suffixes of length $k-l$ of subsequences in $T_h$.
		Let $T_h^{(L), +}$ be the set of subsequences in $T_h^{(L)}$ whose first element is in $\calR_h^+$, and let $T_h^{(R), -}$ be the set of subsequences in $T_h^{(R)}$ whose last element is in $\calR_h^-$. Let $U$ be the set of $l$-th elements of sequences in $T_h^{(L), +}$; note that $\left|T_h^{(L), +}\right|, \left|T_h^{(L), +}\right| = |T_h| / 2 \ge \delta |I_h| / 2$. Let $v$ be the $(|U|/2)$-th larget element of the multiset $\{f(u) : u \in U\}$. Let $U^+$ be the set of elements $u \in U$ with $f(u) \ge f(v)$, and let $U^-$ be the set of elements $u \in U$ with $f(u) \le f(v)$; note that $|U^+|, |U^-| \ge |U| / 2 \ge \delta |I_h| / 4$.
		Thus, with probability at least $1 - \mu / 8$, an element $w$ from $U^+$ is sampled in the second part of Round \ref{round:five}.
		Note that the interval $(x, z)$ contains $\calL_h^+$, and the collection of intervals in $T_h^{(L), +}$ whose $l$-th element is in $U^-$ is a collection of at least $\delta |I_h| / 4 \ge \delta (z - x) / 4$ disjoint $(12 \ldots l)$-patterns whose elements lie in $(x, z) \cap f^{-1}(R)$.  By induction, we find that $\FindMon(f,\, \delta/4,\, l,\, (x, z),\, R \cap (-\infty, f(w)],\, \mu / (16))$ outputs a length-$l$ increasing subsequence whose elements are in $(x, z) \cap f^{-1}(R) \cap f^{-1}((-\infty, f(y)])$.  Similarly, $T_h^{(R),-}$ is a collection of at least $\delta |I_h| / 2 \ge \delta(x + 2^t - z) / 4$ length-$(k-l)$ increasing subsequences whose elements are in $(z, x + 2^t) \cap f^{-1}(R) \cap f^{-1}((f(y), \infty))$.
		Again, by induction, $\FindMon(f,\, \delta/4,\, k-l,\, (z, x + 2^t),\, R \cap (-\infty, f(v)],\, \mu / (16))$ yields a length-$(k-l)$ increasing subsequence whose elements are in $(z, x + 2^t) \cap f^{-1}(R) \cap f^{-1}((f(y), \infty))$. Putting everything together, it follows that, with probability at least $1 - \mu$, a length-$k$ increasing subsequence is output at the last part of Round \ref{round:five}. 
	\end{proof}
}
			
\bibliographystyle{alpha}
\bibliography{waingarten}				

\end{document}